\documentclass[12pt,a4paper]{article}

\usepackage[english]{babel}
\usepackage[utf8x]{inputenc}
\usepackage[T1]{fontenc}
\usepackage{geometry}
\usepackage{amsmath,amsthm,amsfonts}
\usepackage{amssymb}
\usepackage{mathtools}
\usepackage{graphicx}
\usepackage[colorlinks=true]{hyperref}
\newtheorem{theorem}{Theorem}[section]
\newtheorem{lemma}[theorem]{Lemma}
\newtheorem{proposition}[theorem]{Proposition}

\theoremstyle{definition}
\newtheorem{definition}[theorem]{Definition} %[section]

\newtheorem{remark}[theorem]{Remark}
\newcommand{\C}{\mathbb{C}}
\newcommand{\R}{\mathbb{R}}
\newcommand{\N}{\mathbb{N}}
\newcommand{\Z}{\mathbb{Z}}
\newcommand{\SOw}[1]{CHANGE TO M} %{\mbox{\textbf{S}}_{0,#1}}
\newcommand{\M}[2]{\mbox{M}^{#1}_{#2}}
\newcommand{\Schw}{\mathcal{S}}
\newcommand{\twist}{\,\natural\,}
\newcommand{\tr}{\mathrm{tr}}
\newcommand{\lhs}[2]{\prescript{}{\bullet}\!\langle #1, #2 \rangle}
\newcommand{\rhs}[2]{\langle #1, #2 \rangle_{\! \bullet}}

\newcommand{\energy}{\mbox{E}}
\newcommand{\projection}{\mathcal{P}_{s}}

\newcommand{\dd}{\mathrm{d}}
\newcommand{\nn}{\nonumber}

\title{Solitons of general topological charge \\ over noncommutative tori}

\author{Ludwik Dabrowski
\thanks{SISSA (Scuola Internazionale Superiore di Studi Avanzati),
via Bonomea 265, 34136 Trieste, Italy, \mbox{E-mail: \protect\url{dabrow@sissa.it}}}, Mads S.\ Jakobsen\thanks{Norwegian University of Science and Technology, Department of Mathematical Sciences, Trondheim, Norway, \mbox{E-mail: \protect\url{mads.jakobsen@ntnu.no}; \protect\url{franz.luef@ntnu.no}}}, Giovanni Landi\thanks{Matematica, Universita di Trieste, Via A.~Valerio~12/1, I-34127 Trieste, Italy,
and INFN, Sezione di Trieste, Trieste, Italy, \mbox{E-mail: \protect\url{landi@units.it}}}, Franz Luef\footnotemark[2]}

\begin{document}
\maketitle
\begin{abstract}
We continue the study of solitons over noncommutative tori from the perspective of time-frequency analysis and treat the case of general topological charge. Solutions are associated with vector bundles of higher rank over noncommutative tori. We express these vector bundles in terms of vector-valued Gabor frames and apply the duality theory of Gabor analysis to show that Gaussians are solitons of general topological charge over noncommutative tori. An energy functional for projections over noncommutative tori is the basis for the self and anti-self duality equations of the solitions which turns out to have a reformulation in terms of Gabor atoms and we prove that projections generated by Gaussians minimize this energy functional. Finally we comment on the case of the Moyal plane and the associated continuous vector-valued Gabor frames and show that Gaussians are the only class of solitons.
\end{abstract} 

\setcounter{tocdepth}{1}
\tableofcontents
\parskip = 1 ex

\section{Introduction}

Solitons over noncommutative tori of topological charge one were treated in \cite{dalalu15} via Gabor frames, a well-known object of time-frequency analysis. The equivalence of the construction of projections in noncommutative tori and (tight) Gabor frames \cite{lu11} indicates a potential relation between solitons over noncommutative tori and Gabor frames. 

In this paper we discuss the case of solitons of general topological charge by interpreting the relevant projective modules over noncommutative tori as vector-valued Gabor frames. Furthermore we show that Gaussians are solutions of self-duality equations (or anti-self duality equations) on noncommutative tori, and can be termed (noncommutative) sigma-model solitons. These equations are derived from Euler-Lagrange equations for an energy functionals for projections in noncommutative tori that have been introduced in \cite{dakrla00,dakrla03} and further studied in \cite{ro08-3,maro11,le16-2}. 

The energy functional for projections in noncommutative tori becomes a functional for functions generating Gabor frames in $L^{2}(\R\times\Z_{q})$, were $\Z_{q}$, with $q\in \N$ is the finite cyclic group $\{0, 1, 2, \ldots, q-1\}$ under addition modulo $q$. In our reformulation of higher-rank vector bundles over noncommutative tori in terms of Gabor frames, we also develop some of the results in \cite{dalalu15} in this setting, such as the computation of the topological charge of solitons in terms of the Connes-Chern number of projections in noncommutative tori. Our focus in this investigation is the time-frequency aspects of solitons and so we refer the reader interested into the operator algebra and noncommutative geometry aspects to \cite{dalalu15,la06-3}. We hope that this makes our exposition of this intriguing link between Gabor frames and noncommutative geometry accessible to a wider audience.

For now our results are best explained for $q=1$. In this case we consider Gabor frames for $L^{2}(\R)$ of the familiar form $\{E_{m\beta}T_{n\alpha} g\}_{m,n\in\Z}$, where $g$ is a function in the Schwartz class or, more generally, in the Banach space $\M{1}{s}(\R)$ for some $s\ge 2$, and where $\alpha$ and $\beta$ are parameters in $\R\backslash\{0\}$ such that $\vert \alpha\beta \vert< 1$.
We then show that the energy functional,
\[ E(g) = \frac{\pi}{\vert \alpha\beta\vert} \, \sum_{n,m\in\Z} ((\alpha n)^{2} + (\beta m)^{2} ) \ \vert \langle g, E_{m\beta} T_{n\alpha} S^{-1}_{g} g \rangle \vert^{2} \]
is bounded from below by the constant $q=1$ (here $S_{g}^{-1}$ is the inverse of the Gabor frame operator generated by the above Gabor system) and that the (generalized) Gaussian $g(x) = e^{-\pi x^{2} - i\lambda x}$, $\lambda\in \C$ attains this minimum. A similar result holds for $q\ne 1$. We do not know if there are other Gabor frame generators that obtain this minimum.

In the final section we discuss solitons of general topological charge over the Moyal plane and prove that also in the general case, Gaussians are the only minimizers of the energy functional for projections in the Moyal plane algebra and consequently the only solitons (for our sigma-model) on the Moyal plane.

\section{Subspaces of Feichtinger's algebra}
A space which turns out to be very well suited for our purposes is the Banach space of functions known as $\M{1}{s}(\R)$. This is a weighted modulation space introduced by Feichtinger in the '80s. 
In this section we recall some facts about it (see e.g. \cite{gr01} and \cite{fe06}). 

In the following we let $(T_{x})g(t) = g(t-x)$, $x,t\in\R$ be the translation operator and 
$(E_{\omega})g(t) = e^{2\pi i t \omega}g(t)$, $\omega,t\in\R$ be the modulation operator.

\begin{definition}
Fix any function $g$ in the Schwartz space $\Schw(\R)$ and let $s$ be a non negative real number. 
%$s\ge 0$ . 
The \emph{weighted modulation space of order} $s$ is defined to be,
\[\M{1}{s}  (\R) = \big\{ f \in L^{2}(\R) \, : \, \int_{\R^{2}} \vert \langle f, E_{\omega} T_{x} g \rangle \vert \, (1+\vert x \vert + \vert \omega \vert)^{s} \ \dd(x,\omega) < \infty\big\}.\]
\end{definition}
\noindent
For $s=0$ the space $\M{1}{0}(\R)$ is the Feichtinger algebra. 
The norm 
\[ \Vert f \Vert_{\M{1}{s},g } = \int_{\R^{2}} \vert \langle f, E_{\omega} T_{x} g \rangle \vert \, (1+\vert x \vert + \vert \omega \vert)^{s} \ \dd(x,\omega) \]
turns $\M{1}{s}(\R)$ into a Banach space. One can show that different choices of $g$ define the same space and moreover yield equivalent norms. 
Let $\mathcal{F}$ denote the Fourier transform. If $0 \le s_{1} < s_{2}$, there are 
dense and continuous inclusions
\[ \Schw(\R) = \bigcap\limits_{s\in\N} \M{1}{s}(\R) \subseteq \M{1}{s_{2}}(\R) \subsetneq \M{1}{s_{1}}(\R) \subsetneq L^{1}(\R)\cap \mathcal{F}L^{1}(\R) \subsetneq C_{0}(\R).\]
Furthermore $\M{1}{s}(\R)$, $s\ge 0$, is dense in and continuously embedded into $L^{2}(\R)$.
The translation and modulation operators $T_{x}$ and $E_{\omega}$ are bounded on $\M{1}{s}(\R)$, where the operator norm depends on the order $s$ and on $x$ and $\omega$, respectively.  
%the inclusion of $\M{1}{s_{2}}(\R)$ in $\M{1}{s_{1}}(\R)$ is continuous. Moreover, $\M{1}{}(\R)$ is dense in and continuously embedded into $L^{1}(\R)$, $\mathcal{F}L^{1}(\R)$, $C_{0}(\R)$.

%The Schwartz space $\mathcal{S}(\R)$ and all the spaces $\M{1}{s}(\R)$, $s>0$ are dense in $\M{1}{}(\R)$ and, 
%For any $s\ge 0$ the Fourier transform
%\[ \mathcal{F}: \M{1}{s}(\R)\to\M{1}{s}(\R) , \ \mathcal{F}f(\omega) = \int_{\R} f(t) \, e^{-2\pi i t\omega} \, \dd t  \]
%is a well-defined, linear, bounded and bijective operator from $\M{1}{s}(\R)$ onto itself.

It is well known that the differential operator and multiplication by polynomials map the Schwartz space $\Schw(\R)$ into itself. A similar result holds for the spaces $\M{1}{s}(\R)$.
\begin{proposition} \label{pr:der-and-mul-on-M} For any $s\ge 0$, the operators
\begin{align*}
& D : \M{1}{s+1}(\R) \to \M{1}{s}(\R), \quad D f(t) = \tfrac{d}{dt}f(t) \\
& M : \M{1}{s+1}(\R)\to \M{1}{s}(\R), \quad M f(t) = t\cdot f(t),
\end{align*}
are well-defined, linear and bounded. Moreover, 
\[ \mathcal{F} D f = 2\pi i \, M \mathcal{F} f , \qquad \text{for all} \ \ f\in \M{1}{s}(\R), \ s\ge 1.\]
\end{proposition}
\begin{proof} It is straightforward to verify that
\[ \langle M f, E_{\omega} T_{x} g\rangle = \langle f, E_{\omega}T_{x} Mg\rangle + x \, \langle f,E_{\omega}T_{x} g \rangle \]
and, by use of partial integration,
\[ \langle Df, E_{\omega}T_{x}g\rangle = 2\pi i \omega \, \langle f, E_{\omega} T_{x}g\rangle + \langle f, E_{\omega}T_{x}Dg\rangle. \]
For the operator $D$, for $f\in \M{1}{s}(\R)$, 
\begin{multline*} 
  \int_{\R^{2}} \vert \langle Df, E_{\omega}T_{x}g\rangle \vert \ (1+\vert x\vert +\vert \omega\vert)^{s} \, \dd(x,\omega
) \\
  \le 2\pi \int_{\R^{2}} \vert \omega \vert \, \vert \langle f, E_{\omega}T_{x}g\rangle \vert \ (1+\vert x\vert +\vert \omega\vert)^{s} \, \dd(x,\omega) + \int_{\R^{2}} \vert \langle f, E_{\omega}T_{x} D g\rangle \vert \ (1+\vert x\vert +\vert \omega\vert)^{s} \, \dd(x,\omega) \\
  \le 2\pi \int_{\R^{2}} \vert \langle f, E_{\omega}T_{x}g\rangle \vert \ (1+\vert x\vert +\vert \omega\vert)^{s+1} \, \dd(x,\omega) + \int_{\R^{2}} \vert \langle f, E_{\omega}T_{x} D g\rangle \vert \ (1+\vert x\vert +\vert \omega\vert)^{s+1} \, \dd(x,\omega). 
\end{multline*}
Since the Schwartz functions $g$ and $Dg$ induce equivalent norms on $\M{1}{s+1}(\R)$, we conclude that there is a constant $C>0$ such that
\[ \Vert Df\Vert_{\M{1}{s}} = \int_{\R^{2}} \vert \langle Df, E_{\omega}T_{x}g\rangle \vert \ (1+\vert x\vert +\vert \omega\vert)^{s} \, \dd(x,\omega
) \le C \, \Vert f \Vert_{\M{1}{s+1}}.\]
Similarly, for the operator $M$,
\begin{multline*}
 \int_{\R^{2}} \vert \langle Mf, E_{\omega}T_{x}g\rangle \vert \ (1+\vert x\vert +\vert \omega\vert)^{s} \, \dd(x,\omega) \\
 \le \int_{\R^{2}} \vert \langle f, E_{\omega}T_{x}Mg\rangle \vert \ (1+\vert x\vert +\vert \omega\vert)^{s} \, \dd(x,\omega) + \int_{\R^{2}} \vert x \vert \, \vert \langle f, E_{\omega}T_{x}g\rangle \vert \ (1+\vert x\vert +\vert \omega\vert)^{s} \, \dd(x,\omega) \\
\le \int_{\R^{2}} \vert \langle f, E_{\omega}T_{x}Mg\rangle \vert \ (1+\vert x\vert +\vert \omega\vert)^{s+1} \, \dd(x,\omega) + \int_{\R^{2}} \vert \langle f, E_{\omega}T_{x}g\rangle \vert \ (1+\vert x\vert +\vert \omega\vert)^{s+1} \, \dd(x,\omega) .
\end{multline*}
As before, the Schwartz functions $Mg$ and $g$ induce equivalent norms on $\M{1}{s+1}(\R)$ and so, for some $C>0$, we have that
\[ \Vert f \Vert_{\M{1}{s}} = \int_{\R^{2}} \vert \langle Mf, E_{\omega}T_{x}g\rangle \vert \ (1+\vert x\vert +\vert \omega\vert)^{s} \, \dd(x,\omega
) \le C \, \Vert f \Vert_{\M{1}{s+1}}.\]
Lastly, by use of partial integration, we establish that
\begin{align*} \mathcal{F}Df(\omega) & = \int_{\R} \Big( \tfrac{d}{dt}f(t) \Big) \, e^{-2\pi i t \omega } \, \dd t = f(t) e^{-2\pi i t\omega } \Big\vert_{t=+\infty}^{-\infty} - \int_{\R} f(t) \, (-2\pi i \omega) e^{-2 \pi i t \omega } \, \dd t \\ 
& = 2\pi i \omega \int_{\R} f(t) \, e^{-2\pi i t \omega } \, \dd t = 2\pi i M \mathcal{F}f(\omega) \ \ \text{for all} \ \ f\in \M{1}{s}(\R), \ s\ge 1. \end{align*}
This concludes the proof.
\end{proof} 

%From now on, we consider functions on $\R\times\Z_{q}$, where $\Z_q$ denotes the finite abelian cyclic group of order $q$. The space $\M{1}{s}(\R\times\Z_{q})$ is defined as follows: 
%\[ \M{1}{s}(\R\times\Z_{q}) = \{ f \in L^{2}(\R\times\Z_{q}) \, : \, f(\,\cdot\,,k) \in \M{1}{s}(\R) \ \text{for all} \ k\in \Z_{q} \}, \]
%endowed with a norm $\M{1}{s}(\R\times\Z_{q})$ given by
%\[ \Vert f \Vert_{\M{1}{s}(\R\times\Z_{q})} = \sum_{k\in \Z_{q}} \Vert f(\,\cdot\, , k) \Vert_{\M{1}{s}(\R)}. \]
%Notice that there is no weight in the finite component in $\Z_{q}$.
 
%In fact, all weights in the discrete and finite direction (or more general on a compact group) are equivalent to a constant weight of value $1$.

\section{Gabor frames and non-commutative tori}

We need to review some theory on Gabor analysis and non-commutative tori. 
%$\Z_{q}$, $q\in \N$ is the finite cyclic group of integers $\{0,1,2,\ldots,q-1\}$ under addition modulo $q$. %For a function evaluation at a point $(x,j)\in \R\times \Z_{q}$ we write $f(x,j)$. 
%Alternatively, one may see an element $f\in L^{2}(\R\times Z_{q})$ as a collection of functions $f_{j}\in L^{2}(\R)$, $j\in \Z_{q}$. For $(x,j)\in \R\times \Z_{q}$, we thus write $f(x,j) = f_{j}(x)$. 

Consider the space $\R\times\Z_{q}$, where $\Z_q$ is the finite abelian cyclic group of order $q$.
We first define the translation and modulation operators on function on $\R\times\Z_{q}$.

\noindent 
For every $(\lambda,l)\in \R\times \Z_{q}$ we define the \emph{translation operator} (time shift) as
\[ 
T_{\lambda,l} : L^{2}(\R\times\Z_{q}) \to L^{2}(\R\times\Z_{q}), \ (T_{\lambda,l} f)(x,j) = f(x-\lambda,j-l).
\]
For every $(\gamma,c)\in \widehat{\R}\times \widehat{\Z}_{q}$ (the $\widehat{\phantom{m}}$ above the $\R$ and $\Z_{q}$ indicate that $\gamma\in \R$ and $c\in \Z_{q}$ are variables in the frequency domain) we define the \emph{modulation operator} (frequency shift)
\[ E_{\gamma,c} : L^{2}(\R\times\Z_{q}) \to L^{2}(\R\times\Z_{q}), \ (E_{\gamma,c} f)(x,j) = e^{2\pi i (x\cdot \gamma + j c/q)} f(x,j).\]
For $\nu=(\lambda,l,\gamma,c)\in \R\times\Z_{q}\times\widehat{\R}\times\widehat{\Z}_{q}$, 
we then have the \emph{time-frequency shift operator}
\begin{align*} 
 & \pi(\nu) : L^{2}(\R\times\Z_{q}) \to L^{2}(\R\times\Z_{q}), \\ 
 & \pi(\nu) f(x,j)= (E_{\gamma,c}T_{\lambda,l} f) (x,j) = e^{2\pi i (x\cdot \gamma + j c /q)}f(x-\lambda,j-l). 
\end{align*} 
Observe that the time and frequency shift operators commute up to a phase factor, 
$E_{\gamma,c} T_{\lambda,l} = e^{-2\pi i (\lambda\cdot\gamma+lc/q)} T_{\lambda,l}E_{\gamma,c}$. This phase factor is \emph{irrelevant} in the theory of Gabor frames. However, for the theory of the non-commutative torus this phase factor is paramount. Because of this, we also introduce the \emph{frequency-time shift operator}
\begin{align*} 
& \pi^{\circ}(\nu) : L^{2}(\R\times\Z_{q}) \to L^{2}(\R\times\Z_{q}), \\
& \pi^{\circ}(\nu) f(x,j) = T_{\lambda,l} E_{\gamma,c} f(x,j) = e^{2\pi i ((x-\lambda)\cdot \gamma + (j-l) c /q)}f(x-\lambda,j-l).\end{align*}
For $\nu_{1} = (\lambda_{1},l_{1},\gamma_{1},c_{1})$ and $\nu_{2} = (\lambda_{2},l_{2},\gamma_{2},c_{2})$ it is useful to define the $2$-cocycle 
\[ 
\varphi(\nu_{1},\nu_{2}) = e^{-2\pi i (\lambda_{1}\gamma_{2} + l_{1}c_{2}/q)}.
\]
Note that $\varphi(-\nu_{1},\nu_{2}) = \varphi(\nu_{1},-\nu_{2}) = \overline{\varphi(\nu_{1},\nu_{2})}$ (this $2$-cocycle is co-homologous to the anti-symmetrised one
$ 
\varphi ' (\nu_{1},\nu_{2}) = e^{-\pi i \big( (\lambda_{1}\gamma_{2} - \lambda_{2}\gamma_{1}) 
+ (l_{1}c_{2} - l_{2}c_{1}) /q \big)}
$).
Furthermore, some little algebra shows that
\begin{align*} & \pi(\nu) = \overline{\varphi(\nu,\nu)} \, \pi^{\circ} (\nu), \\ 
& \pi(\nu_{1})\pi(\nu_{2}) = \varphi(\nu_{1},\nu_{2})\, \pi(\nu_{1}+\nu_{2}) \ \quad \text{and} \quad \ \pi(\nu)^{*} = \pi^{\circ}(-\nu) = \varphi(\nu,\nu) \, \pi(-\nu) , \\
& \pi^{\circ}(\nu_{1})\pi^{\circ}(\nu_{2}) = \overline{\varphi(\nu_{2},\nu_{1})} \, \pi^{\circ}(\nu_{1}+\nu_{2}) \ \quad \text{and} \quad \ (\pi^{\circ}(\nu))^{*} =  \pi(-\nu) = \overline{\varphi(\nu,\nu)} \, \pi^{\circ}(-\nu). \end{align*}

The space $\R\times\Z_{q}\times\widehat{\R}\times\widehat{\Z}_{q}$ is the \emph{time-frequency plane} or \emph{phase space}. The real line $\R$ and its frequency domain $\widehat{\R}$ are equipped with the usual Lebesgue measure. The group $\Z_{q}$ is equipped with the counting measure, whereas $\widehat{\Z}_{q}$ is equipped with the counting measure times $q^{-1}$.

%These normalizations ensure the Fourier inversion formula and the Plancherel theorem for functions in $L^{2}(\R\times \Z_{q})$. We thus have, for functions $f\in L^{1}(\R\times\Z_{q})$, the Fourier transform \[ \hat{f}(\omega,k) = \int_{\R} \sum_{j\in \Z_{q}} f(x,j) e^{-2\pi i (x\cdot \omega + j k /q)} \, dx, \ (\omega,k)\in \widehat{\R}\times\widehat{\Z}_{q} \] and, for $f\in L^{1}(\R\times \Z_{q})$  with $\hat{f}\in L^{1}(\widehat{\R}\times \widehat{\Z}_{q})$, \[ f(x,j) = \frac{1}{q} \int_{\R} \sum_{k\in \Z_{q}} \hat{f}(\omega,k) e^{2\pi i (x\cdot \omega + j k /q)} \, d\omega, \ (x,j)\in \R\times \Z_{q}.\]

In parallel with the \emph{weighted modulation space of order} $s\geq 0$ introduced in the previous section, the space $\M{1}{s}(\R\times\Z_{q})$ is defined as follows: 
\[ \M{1}{s}(\R\times\Z_{q}) = \{ f \in L^{2}(\R\times\Z_{q}) \, : \, f(\,\cdot\,,k) \in \M{1}{s}(\R) \ \text{for all} \ k\in \Z_{q} \}, \]
endowed with a norm $\M{1}{s}(\R\times\Z_{q})$ given by
\[ \Vert f \Vert_{\M{1}{s}(\R\times\Z_{q})} = \sum_{k\in \Z_{q}} \Vert f(\,\cdot\, , k) \Vert_{\M{1}{s}(\R)}. \]
Notice that there is no weight in the finite component in $\Z_{q}$.
We also need the space
$\M{1}{s}(\R^{2})$ which consists of all functions $F \in L^{2}(\R^{2})$ that satisfy
\[ \int_{\R^{4}} \vert \langle F , E_{\omega_{1},\omega_{2}} T_{x_{1},x_{2}} G \rangle \vert ( 1 + \vert x_{1} \vert + \vert x_{2} \vert + \vert \omega_{1} \vert + \vert \omega_{2} \vert )^{s} \, d(x_{1},x_{2},\omega_{1},\omega_{2}) 
< \infty
 \]
for some fixed non-zero function $G\in\mathcal{S}(\R^{2})$. And the space 
\begin{multline*} 
\M{1}{s}(\R\times\Z_{q}\times\widehat{\R}\times\widehat{\Z}_{q}) \\
= \{ F \in L^{2}(\R\times\Z_{q}\times\widehat{\R}\times\widehat{\Z}_{q}) \, : \, F(\,\cdot\,, k, \,\cdot\,,l) \in \M{1}{s}(\R^{2}) \ \ \mbox{for all} \ \ (k,l)\in \Z_{q}\times\widehat{\Z}_{q}\}. 
\end{multline*}

\begin{definition}
Given two functions $f,g\in \M{1}{s}(\R\times\Z_{q})$ we define the  \emph{short-time Fourier transform} of $f$ with respect to $g$ to be the function 
\[ \mathcal{V}_{g}f : \R\times\Z_{q}\times\widehat{\R}\times\widehat{\Z}_{q} \to \C, \quad \mathcal{V}_{g}f(\lambda,l,\gamma,c) = \langle f, E_{\gamma,c}T_{\lambda,l}g\rangle.\]
\end{definition}
\noindent
Then, if $f,g\in\M{1}{s}(\R\times\Z_{q})$, one can show that 
$\mathcal{V}_{g}f \in \M{1}{s}(\R\times\Z_{q}\times\widehat{\R}\times\widehat{\Z}_{q})$.

In order to define the Gabor systems that we will be working with, and equivalently the non-commutative tori that we will be considering, we need to define lattices in the time-frequency plane $\R\times\Z_{q}\times\widehat{\R}\times\widehat{\Z}_{q}$. Then, let $r$ and $s$ be some integers in $\{ 1, 2, \ldots, q-1\}$ that are co-prime to $q$ (if $q=1$ take $r=s=0$) and let $\alpha$ and $\beta$ be two non-zero real parameters. In the time domain we then define the lattice 
\[ \Lambda = \{ (\lambda,l) \in \R\times\Z_{q} \, : \, \lambda = \alpha n,\, l = rn \mbox{ mod } q 
\quad \mbox{ for all } n\in \Z\}.\]
Whereas in the frequency domain we consider the lattice
\[ \Gamma = \{ (\gamma,c) \in \widehat{\R}\times\widehat{\Z}_{q} \, : \, \gamma = \beta m,\, c = sm \mbox{ mod } q \quad \mbox{ for all } m\in\Z\}. \]

With the normalizations of the measures as described above we find that the measure of a fundamental domain of $\Lambda$ is $\mu(\Lambda) = q \vert \alpha \vert$, whereas $\mu(\Gamma) = \vert \beta \vert$. Note that $\mu(\Lambda)$ does not depend on $r$, nor does $\mu(\Gamma)$ depend on $s$. 
If $q=1$ then we recover the usual lattices $\Lambda = \alpha\Z$ and $\Gamma=\beta\Z$ used in Gabor analysis on $L^{2}(\R)$.
%and the usual non-commutative torus generated by the two unitaries $E_{\beta}$ and $T_{\alpha}$.
\begin{remark} \label{ex:connes1} The lattices $\Lambda$ and $\Gamma$ considered here are inspired by \cite{co80}. There, for some $a\in \R$, one takes
\[ \alpha = a-r/q, \ \beta=1 \ \text{and} \ s=1,\]
such that $\alpha\ne 0$.
\end{remark}

In the following paragraphs we detail how these lattices $\Lambda$ and $\Gamma$ are used to construct the non-commutative tori and the Gabor systems that we are interested in. 

Let us consider the space of weighted $\ell^{1}$-sequences indexed by the lattice $\Lambda\times\Gamma$ of the time-frequency plane. That is we define  
\[ 
\ell^{1}_{s}(\Lambda\times\Gamma) = \Big\{ a \in \ell^{1}(\Lambda\times\Gamma) \, : \, \sum_{(\lambda,l,\gamma,c)\in \Lambda\times\Gamma} \vert a(\lambda,l,\gamma,c) \vert \, (1+\vert \lambda\vert + \vert \gamma\vert)^{s} <\infty \Big\}.
\]
This vector space becomes an involutive Banach algebra under the norm
\[ \Vert a \Vert_{\ell^{1}_{s}} = \sum_{(\lambda,l,\gamma,c)\in \Lambda\times\Gamma} \vert a(\lambda,l,\gamma,c) \vert \, (1+\vert \lambda\vert + \vert \gamma\vert)^{s} \ \quad \text{for all} \ \ a\in \ell^{1}_{s}(\Lambda\times\Gamma),\]
the twisted involution
\[ 
\phantom{m}^{*} : \ell^{1}_{s}(\Lambda\times\Gamma) \to \ell^{1}_{s}(\Lambda\times\Gamma), \quad 
(a^{*})(\nu) = \varphi(\nu,\nu) \overline{a(-\nu)} \ \quad \text{for all} \ \ \nu\in\Lambda\times\Gamma,\]
and with respect to the twisted convolution
\begin{align}
& \twist : \ell^{1}_{s}(\Lambda\times\Gamma)\times \ell^{1}_{s}(\Lambda\times\Gamma) \to \ell^{1}_{s}(\Lambda\times\Gamma), \nn \\ (a_{1} &\twist a_{2})(\nu) = \sum_{\nu'\in \Lambda\times\Gamma} a_{1}(\nu') \, a_{2}(\nu-\nu') \, \varphi(\nu',\nu-\nu'). 
\end{align}
One can show that the map
\[ 
I : a \mapsto \sum_{\nu\in\Lambda\times\Gamma} a(\nu) \pi(\nu), \quad a \in \ell^{1}_{s}(\Lambda\times\Gamma) 
\]
is an isometric isomorphism from $\ell^{1}_{s}(\Lambda\times\Gamma)$ onto the involutive Banach algebra 
\begin{align*} & \mathcal{A}_{s}  = \Big\{ T:\M{1}{s}(\R\times\Z_{q})\to \M{1}{s}(\R\times\Z_{q}) \, : \, T= \sum_{\nu\in\Lambda\times\Gamma} a(\nu) \pi(\nu) , \ 
 a\in \ell^{1}_{s}(\Lambda\times\Gamma) \Big\}.
%&  \mathcal{A}^{\circ}_{s} = \Big\{ T:\M{1}{s}(\R\times\Z_{q})\to \M{1}{s}(\R\times\Z_{q}) \, : \, T= \sum_{\nu^{\circ}\in\Gamma^{\perp}\times\Lambda^{\perp}} b(\nu^{\circ}) \pi^{\circ}(\nu^{\circ}) , \ 
% b\in \ell^{1}_{s}(\Gamma^{\perp}\times\Lambda^{\perp}) \Big\}. 
\end{align*}
It is clear that all elements in $\mathcal{A}_{s}$ are linear and bounded operators on $L^{2}(\R\times\Z_{q})$.
Indeed, $\mathcal{A}_{s}$ is an involutive Banach algebra under the norm $\Vert T\Vert_{A_{s}} = \Vert a \Vert_{\ell^{1}_{s}}$, the composition of operators, and the involution of $T\in \mathcal{A}_{s}$ being its $L^{2}$-Hilbert space adjoint $T^{*}$. 

The enveloping $C^{*}$-algebra ${A}_{\theta}$ of $\ell^{1}_{s}(\Lambda\times\Gamma)$ is the \emph{non-commutative torus} generated, by the two unitaries $U=E_{\beta,s}$, and $V=T_{\alpha,r}$ satisfying then
\[ 
U V = e^{2\pi i \theta} VU, \ \quad \theta = \alpha\beta+rs/q.
\]
\begin{remark} 
If the parameters $\alpha, \beta, r$ and $s$ are chosen as in Remark \ref{ex:connes1}, then $\theta = a$.
\end{remark}
By a celebrated result of Gr\"ochenig and Leinert in \cite{grle04} it follows that $\mathcal{A}_{s}$ is inverse closed in $A_\theta$. More concretely, if $T\in \mathcal{A}_{s}$ and $T^{-1}\in A_\theta$, then  $T^{-1}\in \mathcal{A}_{s}$.

Since a sequence $a\in \ell^{1}_{s}(\Lambda\times\Gamma)$ corresponds to the operator $I(a) \in \mathcal{A}_{s}$, it is natural to define the (left) action of $a$ on a function $f\in \M{1}{s}(\R\times\Z_{q})$ as
\begin{equation} \label{eq:left-action} a \cdot f = I(a) f = \sum_{\nu\in\Lambda\times\Gamma} a(\nu) \pi(\nu) f.\end{equation}
We next construct an $\ell^{1}(\Lambda\times\Gamma)$-valued inner-product on $\M{1}{s}(\R\times\Z_{q})$ in the following way.
\begin{lemma} \label{le:1702a} For any $s\ge 0$ the operator
\begin{align*} 
\lhs{\cdot}{\cdot} : \M{1}{s}(\R\times\Z_{q})\times \M{1}{s}(\R\times\Z_{q}) \to \ell^{1}_{s}(\Lambda\times\Gamma),
%\cong\mathcal{A}_{s}
\quad \lhs{f}{g} = \mathcal{V}_{g}f \big\vert_{\Lambda\times\Gamma}
%\cong I( \mathcal{V}_{g}f\big\vert_{\Lambda\times\Gamma} ) 
 \end{align*}
is well-defined. Moreover, there exists a constant $C>0$ such that
 \[ \Vert \lhs{f}{g} \Vert_{\ell^{1}_{s}} \le C \, \Vert f \Vert_{\M{1}{s}} \, \Vert g \Vert_{\M{1}{s}}.  \]
\end{lemma}
\begin{proof} As mentioned earlier, if $f,g\in \M{1}{s}(\R\times\Z_{q})$, then $\mathcal{V}_{g}f\in \M{1}{s}(\R\times\Z_{q}\times\widehat{\R}\times\widehat{\Z}_{q})$. In fact,
for some $C>0$, one has $\Vert \mathcal{V}_{g}f \Vert_{\M{1}{s}} \le C \, \Vert f \Vert_{\M{1}{s}} \, \Vert g \Vert_{\M{1}{s}}$.
Furthermore, if $\Lambda\times\Gamma$ is a discrete subgroup of $\R\times \Z_{q}\times \widehat{\R} \times \widehat{\Z}_{q}$, then the restriction operator
\[ R_{\Lambda\times\Gamma} : \M{1}{s}(\R\times \Z_{q}\times \widehat{\R} \times \widehat{\Z}_{q}) \to \ell^{1}_{s}(\Lambda\times\Gamma), \quad R_{\Lambda\times\Gamma} F(\nu) = F(\nu), \ \ \nu\in\Lambda\times\Gamma \]
is linear and bounded. Hence $\mathcal{V}_{g}f\big\vert_{\Lambda\times\Gamma}$ is a sequence in $\ell^{1}_{s}(\Lambda\times\Gamma)$ and the norm estimates follow. \end{proof}

One can show that the inner-product is compatible with the action defined in \eqref{eq:left-action} and the twisted convolution and involution on $\ell^{1}_{s}(\Lambda\times\Gamma)$ given above. 
That is, for all $f,g\in\M{1}{s}(\R\times\Z_{q})$ and $a\in \ell^{1}_{s}(\Lambda\times\Gamma)$,
\begin{align} 
\label{eq:inner-p-rule} 
a \twist \lhs{f}{g} &= \lhs{a \cdot f}{g} \ , \quad \lhs{f}{g}\twist a^{*} = \lhs{f}{a \cdot g} \ , 
\quad (\lhs{f}{g})^{*} = \lhs{g}{f} , \nn \\ ~\nn \\
\lhs{f}{f} & \ge 0 \ \ \text{and} \ \ \lhs{f}{f} = 0 \ \ \Leftrightarrow \ \ f = 0.
\end{align}

Now, a function $g\in L^{2}(\R\times\Z_{q})$ is said to generate a \emph{Gabor frame} for $L^{2}(\R\times\Z_{q})$ if there exists constants $A,B>0$ such that
\begin{equation} \label{eq:frame} A\, \Vert f \Vert_{2}^{2} \le \sum_{\nu\in\Lambda\times\Gamma} \vert \langle f, \pi(\nu) g \rangle \vert^{2} \le B \, \Vert f \Vert_{2}^{2} \ \ \mbox{for all} \ \ f\in L^{2}(\R\times\Z_{q}).\end{equation}
With such a $g$, the collection of functions $\{\pi(\nu) g\}_{\nu\in\Lambda\times\Gamma}$ is the $\emph{Gabor system}$ generated by $g$ and the lattice $\Lambda\times\Gamma$. From now on we will always assume  $g\in\M{1}{s}(\R\times\Z_{q})$ for some $s\ge 0$. In this case the Gabor frame operator
\[ S_{g} : \M{1}{s}(\R\times\Z_{q})\to \M{1}{s}(\R\times\Z_{q}), \quad S_{g} f = \sum_{\nu\in\Lambda\times\Gamma} \langle f, \pi(\nu) g\rangle \, \pi(\nu) g = \lhs{f}{g} \cdot  g\]
is well-defined, linear and bounded; it is also positive. The lower inequality in \eqref{eq:frame} implies that the frame operator $S_{g}$ is invertible on $L^{2}(\R\times\Z_{q})$. The aforementioned result by Gr\"ochenig and Leinert on the invertibility implies that $S_{g}$ is also invertible on $\M{1}{s}(\R\times\Z_{q})$ if $g\in \M{1}{s}(\R\times\Z_{q})$. In particular, this invertibility allows for series representations of any function $f\in \M{1}{s}(\R\times\Z_{q})$ (in fact, for all functions in $L^{2}(\R\times\Z_{q})$) of the form
\begin{equation} \label{eq:dual0}
f = \sum_{\nu\in\Lambda\times\Gamma} \langle f, \pi(\nu) g \rangle \, \pi(\nu) S^{-1}_{g}g 
= \lhs{f}{g} \cdot S^{-1}_{g}g  \qquad \text{for all} \ \ f\in \M{1}{s}(\R\times\Z_{q}).
\end{equation}
There may be other functions $h\in \M{1}{s}(\R\times\Z_{q})$, $h\ne S^{-1}_{g}g$ such that 
\begin{equation} \label{eq:dual} 
f = \sum_{\nu\in\Lambda\times\Gamma} \langle f, \pi(\nu) g \rangle \, \pi(\nu) h = \lhs{f}{g}\cdot h \qquad \text{for all} \ \ f\in \M{1}{s}(\R\times\Z_{q}).
\end{equation}
In general, if a pair of functions $g$ and $h$ in $\M{1}{s}(\R\times\Z_{q})$ allow for series representations as in \eqref{eq:dual0} and \eqref{eq:dual} we call them a \emph{dual pair} and the pair $g$ and $S^{-1}_{g}g$ is the \emph{canonical} dual pair.
Note that in \eqref{eq:dual} the role of $g$ and $h$ can be interchanged. 

In order to go further with the theory of Gabor frames we need to describe the annihilators $\Lambda^{\perp}$ and $\Gamma^{\perp}$ of the lattices $\Lambda$ and $\Gamma$. The annihilators $\Lambda^{\perp}$ and $\Gamma^{\perp}$ are lattices of the frequency and time domain, respectively,
\begin{align*} \Lambda^{\perp} & = \{ (\xi, \tau) \in \widehat{\R}\times\widehat{\Z}_{q} \, : \, e^{2\pi i(\lambda \xi + l \tau/q)} = 1 \mbox{ for all } (\lambda,l)\in \Lambda \}, \\
 \Gamma^{\perp} & = \{ (\xi, \tau) \in \R\times\Z_{q} \, : \, e^{2\pi i(\gamma \xi + c \tau/q)} = 1 \mbox{ for all } (\gamma,c)\in \Gamma \}. \end{align*}
To conveniently describe these lattices we use the following notation. If $r$ is co-prime to $q$ and $r=\{1,2,\ldots,q-1\}$, we define $r^{\circ}$ to be the unique element in $\{ 1, 2, \ldots, q-1\}$ such that $r r^{\circ} + lq = 1$ for some $l\in\Z$ (that this is possible follows from the Chinese remainder theorem). 
If $q=1$, we take $r=r^{\circ}=0$. If $r=1$, then $r^{\circ} = 1$. If $r=q-1$, then $r^{\circ} = q-1$. Furthermore, $(r^{\circ})^{\circ} = r$. Similarly we define $s^{\circ}$ for the parameter $s$.
\begin{lemma} \label{le:annihilators} 
If $\Lambda$ and $\Gamma$ are as above, then
\begin{align*} & \Lambda^{\perp} = \{ (\xi, \tau) \in \widehat{\R}\times \widehat{\Z}_{q} \, : \, \xi = \frac{n}{\alpha q}, \ \tau = - r^{\circ} n\mbox{ mod }q \quad \mbox{ for all } \, n\in\Z\}, \\ 
 & \Gamma^{\perp} = \{ (\xi, \tau) \in \R\times \Z_{q} \, : \, \xi = \frac{n}{\beta q}, \ \tau = - s^{\circ} 
 n \mbox{ mod }q \quad \mbox{ for all } \, n\in\Z\} \end{align*}
and the measure of a fundamental domain of $\Lambda^{\perp}$ is $\mu(\Lambda^{\perp}) = (q\vert \alpha\vert)^{-1}$, whereas that of a fundamental domain of $\Gamma^{\perp}$ is $\mu(\Gamma^{\perp}) = \vert\beta\vert^{-1}$. 
\end{lemma}
\begin{proof}
For $(\xi,\tau)\in \widehat{\R}\times \widehat{\Z}_{q}$ to be in $\Lambda^{\perp}$ we need, by definition, that
\[ e^{2\pi i (\alpha m \xi +rm \tau/q)} = 1 \qquad \text{for all} \ \ m\in \Z.\]
that is  $\alpha m \xi +rm \tau/q\in \Z$ for all $m\in\Z$,
If $l\in\Z$ is such that  $rr^{\circ}+lq = 1$, then
\[ \frac{1}{q} = \frac{rr^{\circ}}{q} + l.\]
Indeed, 
using this relation it is straightforward to show that for $(\xi,\tau)$ as in the lemma,
\[ \frac{\alpha n m}{\alpha q} - \frac{r r^{\circ} n m }{q} =  {nm} \big( \frac{rr^{\circ}}{q}+l\big) - \frac{rr^{\circ}nm}{q} = nml \in\Z \]
Hence
\[ 
\big\{ 
(\xi, \tau) \in \widehat{\R}\times \widehat{\Z}_{q} \, : \, \xi = \frac{n}{\alpha q}, \ \tau 
= - r^{\circ} n \mbox{ mod }q \quad \mbox{ for all } n\in\Z \big\} \, \subseteq \, \Lambda^{\perp}. \]
In order to show equality we argue as follows: It is a general fact that for a lattice $\Lambda$ and its adjoint $\Lambda^{\perp}$ it holds that $\mu(\Lambda)\mu(\Lambda^{\perp}) = 1$. As remarked earlier, $\mu(\Lambda) = q \, \vert\alpha\vert$. It is not hard to see that the lattice 
\[ 
\big\{ (\xi, \tau) \in \widehat{\R}\times \widehat{\Z}_{q} \, : \, \xi = \frac{n}{\alpha q}, \ \tau = - a_{\Lambda} n\mbox{ mod }q \mbox{ for all } n\in\Z \big\} \subseteq \Lambda^{\perp}. 
\]
has size $(q\vert \alpha\vert )^{-1}$. If it was a lattice strictly contained in $\Lambda^{\perp}$, its size would have to be strictly larger than $(q\vert\alpha\vert)^{-1}$. being this not the case we conclude that it must be $\Lambda^{\perp}$. The calculation for $\Gamma^{\perp}\subseteq \R\times \Z_{q}$ is similar. Note that in order to compute the lattice size, that is, the measure of a fundamental domain of the lattices, it is important to check whether the adjoint lies in the time or in the frequency domain as we have different measures on $\R\times\Z_{q}$ and $\widehat{\R}\times \widehat{\Z}_{q}$ (as described earlier in this section). 
\end{proof}

With the lattice $\Gamma^{\perp}\times \Lambda^{\perp}$ of the time-frequency plane we proceed in the same way as before and consider, for $s\ge 0$, the space of all weighted $\ell^{1}$-sequences, 
\[ \ell^{1}_{s}(\Gamma^{\perp}\times\Lambda^{\perp}) = \Big\{ b \in \ell^{1}(\Gamma^{\perp}\times\Lambda^{\perp}) \, : \, \sum_{(\lambda,l,\gamma,c)\in \Gamma^{\perp}\times\Lambda^{\perp}} \vert b(\lambda,l,\gamma,c) \vert \, (1+\vert \lambda\vert + \vert \gamma\vert)^{s} <\infty \Big\}
\]
and the algebra of operators
\[ \mathcal{A}^{\circ}_{s} = \Big\{ T:\M{1}{s}(\R\times\Z_{q})\to \M{1}{s}(\R\times\Z_{q}) \, : \, T= \sum_{\nu^{\circ}\in\Gamma^{\perp}\times\Lambda^{\perp}} b(\nu^{\circ}) \pi^{\circ}(\nu^{\circ}) , \ 
 b\in \ell^{1}_{s}(\Gamma^{\perp}\times\Lambda^{\perp}) \Big\}. 
\]
Naturally $\mathcal{A}^{\circ}_{s}$ becomes an involutive Banach algebra just as before with $\mathcal{A}_{s}$. The twisted convolution and twisted involution on $\ell^{1}_{s}(\Gamma^{\perp}\times\Lambda^{\perp})$ are defined in such a way that $\mathcal{A}_{s}^{\circ}$ and $\ell^{1}_{s}(\Gamma^{\perp}\times\Lambda^{\perp})$ become isometric isomorphic involutive Banach algebras under the identification
\[ I^{\circ} : \ell^{1}_{s}(\Gamma^{\perp}\times\Lambda^{\perp}) \to \mathcal{A}^{\circ}_{s}, \quad I^{\circ}(b) = \sum_{\nu^{\circ}\in\Gamma^{\perp}\times\Lambda^{\perp}} b(\nu^{\circ}) \pi^{\circ}(\nu^{\circ}). \]
Note that the algebra $\mathcal{A}_{s}^{\circ}$ is generated by frequency-time shifts $\pi^{\circ}$ rather than time-frequency shifts $\pi$ (as was the case with $\mathcal{A}_{s}$). In particular, $\mathcal{A}_{s}^{\circ}$ is generated by the two unitary operators $U^{\circ}=T_{1/\beta q, -s^{\circ}}$ and $V^{\circ}=E_{1/\alpha q, -r^{\circ}}$.
These generators satisfy
\[ 
U^{\circ}V^{\circ} = e^{2\pi i \theta^{\circ}} V^{\circ}U^{\circ} , \ \quad \theta^{\circ} = r^{\circ}s^{\circ}/q - (\alpha\beta q^2)^{-1}.\]

\begin{remark} 
If the parameters $\alpha,\beta, r$ and $s$ are chosen as in Remark \ref{ex:connes1}, one has $\theta^{\circ} = (l+ar^{\circ})/(r-aq)$, where $l\in\Z$ is such that $rr^{\circ} + lq = 1$. 
\end{remark}
Similarly to what we did for the lattice $\Lambda\times\Gamma$, we define the \emph{right} action of an element $b\in \ell^{1}_{s}(\Gamma^{\perp}\times\Lambda^{\perp})$ on a function $f\in \M{1}{s}(\R\times\Z_{q})$ as 
\begin{equation} \label{eq:right-action} f \cdot b =  I(b)f = \sum_{\nu^{\circ}\in\Gamma^{\perp}\times\Lambda^{\perp}} b(\nu^{\circ}) \pi^{\circ}(\nu^{\circ}) f.\end{equation}
For a function $f$ let $f^{\dagger}$ be the involution $t\mapsto \overline{f(-t)}$. We then define the $\ell^{1}_{s}(\Gamma^{\perp}\times\Lambda^{\perp})$-valued inner-product $\M{1}{s}(\R\times\Z_{q})$
\[ \rhs{\cdot}{\cdot} : \M{1}{s}(\R\times\Z_{q})\times \M{1}{s}(\R\times\Z_{q})\to \ell^{1}_{s}(\Gamma^{\perp}\times\Lambda^{\perp}),  \quad \rhs{f}{g} = (q\vert\alpha\beta\vert)^{-1} \, (\mathcal{V}_{g}{f})^{\dagger}\big\vert_{\Gamma^{\perp}\times\Lambda^{\perp}}.\]
 
Note that this inner-product is linear in the second entry and 
$(\mathcal{V}_{g}{f})^{\dagger}(\nu^{\circ}) = \langle g,\pi^{\circ}(\nu^{\circ})f\rangle$.
Its properties are as in Lemma \ref{le:1702a}
, that is, there exists a constant $C>0$ such that
\[ \Vert \rhs{f}{g} \Vert_{\ell^{1}_{s}} \le C \, \Vert f \Vert_{\M{1}{s}} \, \Vert g \Vert_{\M{1}{s}}.  \]

Furthermore, for all $f,g\in\M{1}{s}(\R\times\Z_{q})$ and $b\in \ell^{1}_{s}(\Gamma^{\perp}\times\Lambda^{\perp})$ the inner-product $\rhs{\cdot}{\cdot}$ satisfies
\begin{align*}
\rhs{f}{g} \twist b & = \rhs{f}{g\cdot b} \ , \quad  b^{*} \twist \rhs{f}{g}  
= \rhs{f\cdot b}{g} \ , \quad (\rhs{f}{g})^{*} = \rhs{g}{f} , \nn \\ ~\nn \\
\rhs{f}{f} & \ge 0 \ \ \text{and} \ \ \rhs{f}{f} = 0 \ \ \Leftrightarrow \ \ f= 0.
\end{align*}
%To summarize, for both time-frequency lattices $\Lambda\times\Gamma$ and $\Gamma^{\perp}\times\Lambda^{\perp}$ we have now define weighted $\ell^{1}$-sequence spaces and established that they are isometric isomorphic to $\mathcal{A}_{s}$ and $\mathcal{A}_{s}^{\circ}$, respectively. Where $\mathcal{A}_{s}$ and $\mathcal{A}_{s}^{\circ}$ are generated by time-frequency shifts $\{\pi(\nu)\}_{\nu\in\Lambda\times\Gamma}$ and $\{\pi^{\circ}(\nu^{\circ})\}_{\nu^{\circ}\in \Gamma^{\perp}\times\Lambda^{\perp}}$, respectively. 

The now established notation allows us to formulate well-known results in Gabor analysis in the following way.

\medskip
The \textbf{fundamental identity of Gabor analysis}. \\
This  states that for all $f,g,h\in \M{1}{s}(\R\times\Z_{q})$ one has that
\begin{equation} \label{eq:janssen} 
\lhs{f}{g} \cdot h = f \cdot \rhs{g}{h} \, ,
\end{equation}
that is 
$$
\sum_{\nu\in\Lambda\times\Gamma} \langle f, \pi(\nu) g \rangle \,  \pi(\nu) h = \frac{1}{q\vert\alpha\beta\vert} \sum_{\nu^{\circ}\in \Gamma^{\perp}\times\Lambda^{\perp}} \langle h, \pi^{\circ}(\nu^{\circ})g\rangle \, \pi^{\circ}(\nu^{\circ})f .
$$
If the involved functions are ``nice'' enough, then \eqref{eq:janssen} follows by an application of the Poisson summation formula. This is the case for functions in the Schwartz space \cite{ri88}, or functions in $\M{1}{s}(\R\times\Z_{q})$ \cite{felu06}. The equality does in general not hold for  arbitrary functions 
$f,g,h\in L^{2}(\R\times\Z_{q})$ (see \cite{gr01} and \cite{jale16-2}).

These statements have as consequence the Morita equivalence of $\mathcal{A}_{s}$ and $\mathcal{A}_{s}^{\circ}$, which extends the result of Luef for $q=1$ in \cite{lu09}. 
\begin{proposition}\label{pr:Morita}
The algebras
$\mathcal{A}_{s}$ and $\mathcal{A}_{s}^{\circ}$ are Morita equivalent and $\M{1}{s}(\R\times\Z_{q})$ is an equivalence bimodule. Consequently, $\M{1}{s}(\R\times\Z_{q})$ is a projective finitely generated $\mathcal{A}_{s}$-module, i.e. there exist finitely many $g_1,...,g_n$ in 
$\M{1}{s}(\R\times\Z_{q})$ such that any $f\in \M{1}{s}(\R\times\Z_{q})$ may be written as 
\[ f=\lhs{f}{g_1}g_1+\cdots+\lhs{f}{g_n}g_n.\] 
\end{proposition}
\noindent
In noncommutative geometry one says that $\M{1}{s}(\R\times\Z_{q})$ is a vector bundle over $\mathcal{A}_s$.
%In the time-frequency community \eqref{eq:janssen} yields a useful way of expressing the Gabor frame operator,
%\[ S_{g} f = f \, \rhs{g}{g}.\]
%In non-commutative geometry the validity of \eqref{eq:janssen} implies that the non-commutative torus $\mathcal{A}_{s}$ and $\mathcal{A}_{s}^{\circ}$ are Morita-equivalent, which is a far deeper.

\medskip
The \textbf{Wexler-Raz biorthogonality relations}. \\
These characterise when two functions $g,h\in \M{1}{s}(\R\times\Z_{q})$ generate dual Gabor frames, that is,  \eqref{eq:dual} holds. This is the case if and only if
\[ \rhs{g}{h} = 1, \ \text{ i.e.,} \ \ \langle h, \pi^{\circ}(\nu^{\circ}) g\rangle = \begin{cases} q \vert\alpha \beta\vert & \nu^{\circ}=0 \\ 0 &  \nu^{\circ} \ne 0 \end{cases} \, , \qquad \text{for all} \ \ \nu^{\circ}\in\Gamma^{\perp}\times\Lambda^{\perp}. \]

As remarked earlier, if a function $g\in \M{1}{s}(\R\times\Z_{q})$ generates a Gabor frame $\{\pi(\nu)g\}_{\nu\in\Lambda\times\Gamma}$ for $L^{2}(\R\times\Z_{q})$, there exists functions $h\in\M{1}{s}(\R\times\Z_{q})$ such that $\rhs{g}{h}=1$ and therefore
\[ f = \lhs{f}{g} \cdot h \qquad \text{for all} \ \ f\in \M{1}{s}(\R\times\Z_{q}).\]
In particular, one can take the canonical dual generator $h=S^{-1}_{g}g = g \cdot (\rhs{g}{g})^{-1}$. 

\medskip
The \textbf{duality principle for Gabor frames}. \\
This states that $\{\pi(\nu)g\}_{\nu\in\Lambda\times\Gamma}$ is a frame for $L^{2}(\R\times\Z_{q})$ if, and only if, the Gabor system $\{\pi^{\circ}(\nu^{\circ})g\}_{\nu^{\circ}\in\Gamma^{\perp}\times\Lambda^{\perp}}$ 
is a Riesz sequence for $L^{2}(\R\times\Z_{q})$. 
That is, there exist positive constants $c_1,c_2$ such that 
$$
c_1 \sum_{\nu^{\circ}\in\Gamma^{\perp}\times\Lambda^{\perp}} | a_{\nu^{\circ}} |^2 \le 
\Big\Vert \sum_{\nu^{\circ}\in\Gamma^{\perp}\times\Lambda^{\perp}} a_{\nu^{\circ}} \, \pi^{\circ}(\nu^{\circ})g \, \Big\Vert^2 
\le c_{2} \sum_{\nu^{\circ}\in\Gamma^{\perp}\times\Lambda^{\perp}} | a_{\nu^{\circ}} |^2
$$
for all sequences $a\in \ell^{2}(\Gamma^{\perp}\times\Lambda^{\perp})$.
For our purposes the importance of the duality principle is that the Riesz sequence property  implies that 
\begin{equation} \label{eq:riesz-rep} 
f = g \cdot \rhs{h}{f} \qquad\text{for all} \ \ f\in W, 
\end{equation}
where $W$ is the closure of $\mbox{span}\{\pi^{\circ}(\nu^{\circ})g\}$ in $\M{1}{s}(\R\times\Z_{q})$ (in fact it holds for functions in the closure of $\mbox{span}\{\pi^{\circ}(\nu^{\circ})g\}$ in $L^{2}(\R\times\Z_{q})$). Note that $W$ cannot contain all of $\M{1}{s}(\R\times\Z_{q})$ for any $s\ge 0$ nor can $W$ contain the entire Schwartz space. One therefore has to be careful when using the equality in \eqref{eq:riesz-rep}. The duality principle for Gabor frames was proven independently in \cite{dalala95}, \cite{ja95} and \cite{rosh97} for Gabor systems in $L^{2}(\R)$. The duality principle for Gabor systems in $L^{2}(\R\times\Z_{q})$ follows from the general duality principle for Gabor systems on locally compact ablelian groups in \cite{jale16-2}. 

\medskip

The following result shows that generators of Gabor frames for $L^{2}(\R)$ can sometimes be used to generate Gabor frames for $L^{2}(\R\times\Z_{q})$.

\begin{lemma}\label{le:from-1-to-q-frames}
Assume that the parameters $\alpha,\beta,r,s$ and $q$ are such that
\[ 
(\alpha\beta q^{2})^{-1} + r^{\circ}s^{\circ} / q \in\Z.
\]
If $\tilde{g}\in \M{1}{s}(\R)$ generates a Gabor frame for $L^{2}(\R)$ w.r.t.\ the lattice $\alpha\Z\times q\beta\Z$, then the function $g\in \M{1}{s}(\R\times\Z_{q})$ given by $g(\cdot, k) = \tilde{g}$, for $k\in\Z_{q}$ generates a Gabor frame for $L^{2}(\R\times\Z_{q})$ with respect to the lattice $\Lambda$ and $\Gamma$.
\end{lemma}
\begin{proof} 
By 
the duality principle of Gabor frames we know that $\{\pi(\nu) g\}_{\nu\in\Lambda\times\Gamma}$ is a Gabor frame for $L^{2}(\R\times\Z_{q})$ if and only
$\{\pi^{\circ}(\nu^{\circ})g\}_{\nu^{\circ}\in\Gamma^{\perp}\times\Lambda^{\perp}}$ is a Riesz sequence for $L^{2}(\R\times\Z_{q})$. This is the case if and only if the bi-infinite matrix 
\[ \big( \big\langle T_{\tilde{n}/\beta q, -s^{\circ}\tilde{n}} \, E_{\tilde{m}/\alpha q, -r^{\circ}\tilde{m}} g, T_{n/\beta q, -s^{\circ}n} \, E_{m/\alpha q, -r^{\circ}m} g  \big\rangle_{L^{2}(\R\times\Z_{q})} \big)_{m,\tilde{m},n,\tilde{n}\in\Z}\]
is invertible as an operator on $\ell^{2}(\Z^{2})$.
Now
\begin{multline*} \big\langle T_{\tilde{n}/\beta q, -s^{\circ}\tilde{n}} \, E_{\tilde{m}/\alpha q, -r^{\circ}\tilde{m}} g, T_{n/\beta q, -s^{\circ}n} \, E_{m/\alpha q, -r^{\circ}m} g  \big\rangle_{L^{2}(\R\times\Z_{q})}  \\
= \big\langle  g, T_{(n-\tilde{n})/\beta q, -s^{\circ}(n-\tilde{n})} \, E_{(m-\tilde{m})/\alpha q, -r^{\circ}(m-\tilde{m})} g  \big\rangle_{L^{2}(\R\times\Z_{q})} , 
\end{multline*}
since the phase factor coming from commuting the translation and modulation operators disappears due to the 
condition on the parameters that $(\alpha\beta q^{2})^{-1} + r^{\circ}s^{\circ}/q \in\Z$.
This shows that the bi-infinite matrix has Laurent structure. 
As can be found in, e.g. \cite{ja96}, the invertibility of such matrices is equivalent to the fact that the function
\[ F(t_{1},t_{2}) = \sum_{m,n\in\Z} \big\langle  g, T_{n/\beta q, -s^{\circ}n} \, E_{m/\alpha q, -r^{\circ}m} g  \big\rangle_{L^{2}(\R\times\Z_{q})} \, e^{2\pi i (mt_{1}+nt_{2})} , \quad  (t_{1},t_{2})\in\R^{2},\]
is bounded away from zero (and finite, which is automatic for functions in $\M{1}{s}(\R\times\Z_{q})$).
Note that
\[ \big\langle  g, T_{n/\beta q, -s^{\circ}n} \, E_{m/\alpha q, -r^{\circ}m} g  \big\rangle_{L^{2}(\R\times\Z_{q})} = \langle \tilde{g}, E_{m/\alpha q} T_{n/\beta q} \tilde{g}\rangle_{L^{2}(\R)} \, \begin{cases} q 
& m\in q\Z, \\ 0 & \text{otherwise}.\end{cases}.\]
This implies that
\[ F(t_{1},t_{2}) = q \sum_{m,n\in\Z} \langle \tilde{g}, E_{m/\alpha} T_{n/\beta q} \tilde{g}\rangle_{L^{2}(\R)} e^{2\pi i (qmt_{1},nt_{2})}.\]
Again by \cite{ja96} this function is bounded away from zero if and only if $\{E_{m\beta q}T_{n\alpha}\tilde{g}\}_{m,n\in\Z}$ is a Gabor frame for $L^{2}(\R)$, which is true by assumption. The result follows.
\end{proof}

For the following sections it is important to observe that functions that generate dual Gabor frames allow us to construct special elements in the algebra $\ell^{1}_{s}(\Lambda\times\Gamma)$.

\begin{lemma} Let $g$ and $h$ be functions in $\M{1}{s}(\R\times\Z_{q})$, $s\ge 0$, such that the Gabor systems $\{\pi(\nu)g\}_{\nu\in\Lambda\times\Gamma}$ and $\{\pi(\nu)h\}_{\nu\in\Lambda\times\Gamma}$ are dual Gabor frames for $L^{2}(\R\times\Z_{q})$, that is \eqref{eq:dual} is satisfied. Then the following holds.
\begin{enumerate}
\item[(i)] The sequence $a=\lhs{g}{h}\in\ell^{1}_{s}(\Lambda\times\Gamma)$ is idempotent, 
that is, $a^{2} = a \twist a = a$.
\item[(ii)] If $h$ is the canonical dual generator, $h = S^{-1}_{g}g$, then $a = \lhs{g}{S^{-1}_{g}g}$ is a projection, that is, $a^{2} = a \twist a = a$ and $a^{*} = a$.
\end{enumerate}
\end{lemma}
\begin{proof} (i). Since $g$ and $h$ are dual generators they satisfy the Wexler-Raz relations, 
$\rhs{g}{h} = 1$ and, equivalently, $\rhs{h}{g}=1$. We thus find
\[ a^{2} = \lhs{g}{h} \twist \lhs{g}{h} \stackrel{\eqref{eq:inner-p-rule}}{=} \lhs{\lhs{g}{h} \cdot g}{h}  \stackrel{\eqref{eq:janssen}}{=} \lhs{g\cdot\rhs{h}{g}}{h} = \lhs{g}{h} = a.\]
(ii). Since $a^{2}=a$ from (i), we only need to show that $a=a^{*}$.
Recall that the inverse frame operator $S^{-1}_{g}$ is self-adjoint and commutes with time-frequency shifts $\{\pi(\nu)\}_{\nu\in\Lambda\times\Gamma}$. Thus,
\[ a^{*} = (\lhs{g}{S^{-1}_{g}g})^{*} \stackrel{\eqref{eq:inner-p-rule}}{=} \lhs{S^{-1}_{g}g}{g} = \{ \langle S^{-1}_{g}g,\pi(\nu)g\rangle\}_{\nu\in\Lambda\times\Gamma} = \{ \langle g,\pi(\nu)S^{-1}_{g}g\rangle\}_{\nu\in\Lambda\times\Gamma} = a.\]
This concludes the proof. 
\end{proof}

Finally, consider the group $\mbox{GL}(\mathcal{A}_{s}^{\circ})$, of all invertible elements in $\mathcal{A}_{s}^{\circ}\cong\ell^{1}_{s}(\Gamma^{\perp}\times\Lambda^{\perp})$. Its elements 
are called \emph{gauge transformations} of $\M{1}{s}(\R\times\Z_{q})$. 
Example of gauge transformations are the Gabor frame operator $S_g$ of a function $g\in \M{1}{s}(\R\times\Z_{q})$ with time-frequency shifts along the lattice $\Lambda \times\Gamma$, its square root $S^{1/2}_{g}$ and also their inverses, $S^{-1}_g$ and $S^{-1/2}_g$.

\begin{lemma} \label{le:0712} Let $T$ be a gauge transformation. 
\begin{enumerate}
\item[(i)] If $g\in \M{1}{s}(\R\times\Z_{q})$ generates a Gabor frame 
$\{\pi(\nu)g\}_{\nu\in\Lambda\times\Gamma}$ for $L^{2}(\R\times\Z_{q})$, then
\[ \lhs{f_{1}}{S_{g}^{-1} f_{2}} = \lhs{Tf_{1}}{S_{Tg}^{-1}Tf_{2}} \qquad  \mbox{for all} \ \ f_{1},f_{2}\in \M{1}{s}(\R\times\Z_{q}).\]
\item[(ii)] Furthermore,
\[ \lhs{f_{1}}{f_{2}} = \lhs{Tf_{1}}{(T^{-1})^{*} f_{2}} \qquad \mbox{for all} \ \ f_{1},f_{2}\in \M{1}{s}(\R\times\Z_{q}). \]
\end{enumerate}
\end{lemma}
\begin{proof} Since $\mathcal{A}_{s}^{\circ}\cong \ell^{1}_{s}(\Gamma^{\perp}\times\Lambda^{\perp})$, we may write $T f = f \, \cdot \, b$ for some unique $b\in \ell^{1}_{s}(\Gamma^{\perp}\times\Lambda^{\perp})$. Using properties of the inner product $\rhs{\, \cdot\, }{\, \cdot\, }$ we find
\begin{align*}
S_{Tg}^{-1} Tf_{2} & = f_{2} \cdot (b \twist (\rhs{g\cdot b}{g\cdot b})^{-1}) \\
& = f_{2} \cdot (b \twist ( b^{*} \twist \rhs{g}{g} \twist b)^{-1} ) \\
& = f_{2} \cdot (b \twist b^{-1} \twist (\rhs{g}{g})^{-1} \twist (b^{*})^{-1} )\\
& = f_{2} \cdot ( (\rhs{g}{g})^{-1} \twist (b^{*})^{-1} ) .
\end{align*} 
Using this and the fact that $\lhs{f_{1}\cdot b}{f_{2}}=\lhs{f_{1}}{f_{2}\twist b^{*}}$ yields the desired equality:
\[ \lhs{Tf_{1}}{S_{Tg}^{-1}Tf_{2}} = \lhs{f_{1}\cdot b}{f_{2} \cdot ( (\rhs{g}{g})^{-1} \twist (b^{*})^{-1} )} = \lhs{f_{1}}{f_{2}\cdot (\rhs{g}{g})^{-1}} = \lhs{f_{1}}{S_{g}^{-1} f_{2}}.\]
This proves (i). The statement in (ii) follows from the fact that $T$ (and thus also $T^{*}$) commutes with time-frequency shifts from the lattice $\Lambda\times\Gamma$:
\[ \lhs{Tf_{1}}{(T^{-1})^{*} f_{2}} = \{\langle T f_{1}, \pi(\nu) (T^{*})^{-1} f_{2}\rangle\}_{\nu\in\Lambda\times\Gamma} = \{\langle f_{1},\pi(\nu) f_{2}\rangle \}_{\nu\in\Lambda\times\Gamma} =\lhs{f_{1}}{f_{2}}.\]
This concludes the proof
\end{proof}

Lemma \ref{le:0712} implies that the canonical dual pair $g$ and $S^{-1}_gg$ generate the same projection as the canonical tight dual window $S^{-1/2}_gg$, that is, 
\[ \lhs{g}{S^{-1}g} = \lhs{S^{-1/2}g}{S^{-1/2}g}.\]

\section{Derivations, connections and curvature}\label{se:dcc}
In this section we shall detail a few concepts of non-commutative geometry related to non-commutative tori.
In analogy to Riemannian geometry we first consider the following \emph{covariant derivatives} 
on the bundle $\M{1}{s+1}(\R\times\Z_{q})$, $s\ge 0$: 
\begin{itemize}
\item[] $\nabla_{1}: \M{1}{s+1}(\R\times\Z_{q}) \to \M{1}{s}(\R\times\Z_{q}), 
\quad  \nabla_{1}f(\, \cdot \, , k) = 2\pi \, i \, M f(\, \cdot \, , k), \quad k\in \Z_{q}$,
\item[] $\nabla_{2}: \M{1}{s+1}(\R\times\Z_{q}) \to \M{1}{s}(\R\times\Z_{q}), \quad
\nabla_{2}f(\, \cdot \,, k) = D f(\, \cdot\, , k), \quad k\in \Z_{q}$.
\end{itemize}
Note that Proposition \ref{pr:der-and-mul-on-M} implies that these operators are well-defined, linear and bounded. Observe also that $\nabla_{1}$ and $\nabla_{2}$ do not have any action in the discrete variable $k$ and that they do not depend on the parameters $\alpha,\beta, r, s$ nor $q$.

It is straightforward to verify that, for all $f\in \M{1}{s}(\R\times\Z_{q})$ with $s\ge 1$,
\begin{align}
\nabla_{1} (E_{\gamma,c}f) & = E_{\gamma,c} \nabla_{1}f, \qquad \nabla_{1} (T_{\lambda,l} f ) 
= 2\pi i \, \lambda (T_{\lambda,l}f )+ T_{\lambda,l} \nabla_{1}f \label{eq:2102a} \\
\nabla_{2} (T_{\lambda,l}f) & = T_{\lambda,l} \nabla_{2}f, \qquad \nabla_{2} E_{\gamma,c} f 
= 2\pi i \, \gamma (E_{\gamma,c}f) + E_{\gamma,c} \nabla_{2}f \label{eq:2102b}
\end{align}
On $\ell^{1}_{s}(\Lambda\times\Gamma)$ we define \emph{derivations} $\partial_{1}$ and $\partial_{2}$ 
as follows. 
For $s\ge 0$, define
%\begin{itemize}
%\item[] $\partial_{1}: \ell^{1}_{s+1}(\Lambda\times\Gamma) \to \ell^{1}_{s}(\Lambda\times\Gamma), \ (\partial_{1}a)(\lambda,l,\gamma,c) = 2\pi i\lambda \, a(\lambda,l,\gamma,c), \ (\lambda,l,\gamma,c)\in \Lambda\times\Gamma$,
%\item[] $\partial_{2}: \ell^{1}_{s+1}(\Lambda\times\Gamma) \to \ell^{1}_{s}(\Lambda\times\Gamma), \ (\partial_{2}a)(\lambda,l,\gamma,c) = 2 \pi i \gamma \, a(\lambda,l,\gamma,c), \ (\lambda,l,\gamma,c)\in \Lambda\times\Gamma.$
%\end{itemize}
$$
\partial_{j}: \ell^{1}_{s+1}(\Lambda\times\Gamma) \to \ell^{1}_{s}(\Lambda\times\Gamma), \quad j=1,2 
$$
with, for $(\lambda,l,\gamma,c)\in \Lambda\times\Gamma$, 
\[
 (\partial_{1}a)(\lambda,l,\gamma,c) = 2\pi i\lambda \, a(\lambda,l,\gamma,c), \qquad
 (\partial_{2}a)(\lambda,l,\gamma,c) = 2 \pi i \gamma \, a(\lambda,l,\gamma,c). 
\]
Similarly, on $\ell^{1}_{s+1}(\Gamma^{\perp}\times\Lambda^{\perp})$ we define
%\begin{itemize}
%\item[] $\partial_{1}^{\circ}: \ell^{1}_{s+1}(\Gamma^{\perp}\times\Lambda^{\perp}) \to \ell^{1}_{s}(\Gamma^{\perp}\times\Lambda^{\perp}), \ (\partial_{1}^{\circ}a)(\lambda,l,\gamma,c) = 2\pi i\lambda \, a(\lambda,l,\gamma,c), \ (\lambda,l,\gamma,c) \in \Gamma^{\perp}\times\Lambda^{\perp}$,
%\item[] $\partial_{2}^{\circ}: \ell^{1}_{s+1}(\Gamma^{\perp}\times\Lambda^{\perp}) \to \ell^{1}_{s}(\Gamma^{\perp}\times\Lambda^{\perp}), \ (\partial_{2}^{\circ}a)(\lambda,l,\gamma,c) = 2\pi i \gamma \, a(\lambda,l,\gamma,c), \ (\lambda,l,\gamma,c)\in \Gamma^{\perp}\times\Lambda^{\perp}.$
%\end{itemize}
$$ 
\partial_{j}^{\circ}: \ell^{1}_{s+1}(\Gamma^{\perp}\times\Lambda^{\perp}) 
\to \ell^{1}_{s}(\Gamma^{\perp}\times\Lambda^{\perp}), 
\quad j=1,2 
$$
with, for $(\lambda,l,\gamma,c) \in \Gamma^{\perp}\times\Lambda^{\perp}$,
$$
(\partial_{1}^{\circ}a)(\lambda,l,\gamma,c) = 2\pi i\lambda \, a(\lambda,l,\gamma,c), \qquad 
(\partial_{2}^{\circ}a)(\lambda,l,\gamma,c) = 2\pi i \gamma \, a(\lambda,l,\gamma,c).
$$
Using the isomorphism between $\ell^{1}_{s}(\Lambda\times\Gamma)$ and $\mathcal{A}_{s}$ and between $\ell^{1}_{s}(\Gamma^{\perp}\times\Lambda^{\perp})$ and $\mathcal{A}_{s}^{\circ}$ the derivations can naturally be defined on $\mathcal{A}_{s}$ and $\mathcal{A}_{s}^{\circ}$, such that
\[ \partial_{j} : \mathcal{A}_{s+1} \to \mathcal{A}_{s} \quad \text{and} \quad \partial_{j}^{\circ} : \mathcal{A}_{s+1}^{\circ} \to \mathcal{A}_{s}^{\circ} \qquad \text{for} \ \ j=1,2.\]
 
\begin{remark} 
Note that the derivations depend on the lattices $\Lambda$ and $\Gamma$.
In other literature, see e.g. \cite{co80,cori87,dalalu15,la06-3}, the sequence spaces $\ell^{1}_{s}$ are not indexed by the lattice $\Lambda\times\Gamma$ but rather by $\Z^{2}$. 
One therefore defines, e.g., for $j=1,2$, 
\[ 
\partial_{j} : \ell^{1}_{s+1}(\Z^{2}) \to \ell^{1}_{s} (\Z^{2}), \quad 
(\partial_{j} a)(n_1,n_2) = 2 \pi i \, n_j a(n_1,n_2), \quad (n_1,n_2)\in\Z^{2} . \]
Hence, in this case, the derivations are independent on $\Lambda$ and $\Gamma$. This discrepancy has no implication on the theory, it is just a matter of normalization. 
\end{remark}

The derivations are well-defined, linear and bounded operators. 
Concerning boundedness one easily verifies that 
\[ \Vert \partial_{1} a \Vert_{\ell^{1}_{s}} = \sum_{(\lambda,l,\gamma,c)\in \Lambda\times\Gamma} \vert 2\pi i \lambda a(\lambda,l,\gamma,c) \vert (1+\vert \lambda \vert +\vert \gamma\vert)^{s} \le 2\pi \, \Vert a \Vert_{\ell^{1}_{s+1}}. \]
And similar estimates can be established for $\partial_{2}, \partial_{1}^{\circ}$ and $\partial_{2}^{\circ}$.

From the definition of the operators $\nabla_{j}$ and $\partial_{j}$, $j=1,2$ and equations \eqref{eq:2102a} and \eqref{eq:2102b} one establishes that the Leibniz rule holds, that is 
\begin{equation} \label{eq:0103a}
\nabla_{j}( a \cdot  f ) = (\partial_{j} a) \cdot f + a\cdot \nabla_{j} f 
\qquad \text{for all} \ \ f\in \M{1}{1}(\R\times\Z_{q}), \ a\in \ell^{1}_{1}(\Lambda\times\Gamma), 
\end{equation}
and that the derivations are compatible with the $\ell^{1}_{s}$-sequence valued inner-product, 
\begin{equation} \label{eq:0103b} 
\lhs{\nabla_{j} f}{g} + \lhs{f}{\nabla_{j} g} = \partial_{j} (\lhs{f}{\nabla_{j} g} ) 
\qquad \text{for all} \ \ f,g\in \M{1}{1}(\R\times\Z_{q}) .
\end{equation}
Combining these two equations, we find that, for all $f,g,h\in \M{1}{1}(\R\times\Z_{q})$ that
\[ \nabla_{j}(\lhs{f}{g}\, h) = \lhs{\langle \nabla_{j} f}{g} \, h 
+ \lhs{f}{\nabla_{j}g}\, h + \lhs{f}{g} \, \nabla_{j} h \qquad j=1,2.\]
Similar statements hold with $\partial^{\circ}_{j}$ and $\rhs{\cdot}{\cdot}$ in stead of 
$\partial_{j}$ and $\lhs{\cdot}{\cdot}$.

As in Riemannian geometry, the curvature of the covariant derivatives is given by
\[ F_{12} = \nabla_{1}\nabla_{2} - \nabla_{2}\nabla_{1},\]
since $\partial_1$ and $\partial_2$ are two commuting derivations. 
It turns out that the curvature is constant. 
\begin{lemma} For any $s\ge 0$ the curvature of the covariant derivatives is given by the linear and bounded operator
\[ F_{12}: \M{1}{s}(\R\times\Z_{q}) \to \M{1}{s}(\R\times\Z_{q}), 
\qquad F_{12} f = -2\pi i \, \mathrm{Id}. \]
\end{lemma}
\begin{proof}
For any $s\ge 0$, from the definition of $\nabla_{j}$, $j=1,2$  it is clear that $F_{12}$ is a linear and bounded operator from $\M{1}{s+2}(\R\times\Z_{q})$ into $\M{1}{s}(\R\times\Z_{q})$. It is straightforward to show that 
$ F_{12} f = -2 \pi \, i\, f.$
Since $\M{1}{s+2}(\R\times\Z_{q})$ is dense in $\M{1}{s}(\R\times\Z_{q})$ one can extend this operator to all of $\M{1}{s}(\R\times\Z_{q})$ and the result follows.
\end{proof}

\section{Traces and the Connes-Chern number}
In this section we introduce the Connes-Chern (classes and) numbers. In order to do this, we first need to talk about traces. 
A trace on $\ell^{1}_{s}(\Lambda\times\Gamma)$ is a linear and bounded functional $\tr: \ell^{1}(\Lambda\times\Gamma)\to \C$ such that
\[ \tr(a_{1} \, \twist \, a_{2}) = \tr(a_{2} \, \twist \, a_{1}) \qquad \text{for all} \ \ a_{1},a_{2}\in \ell^{1}(\Lambda\times\Gamma), \]
\[ \tr(a^{*}\, \twist\, a) \ge 0, \ \ \text{and} \ \ \tr(a^{*}) = \overline{\tr(a)}\]
If $\tr(a^{*} \,\twist\, a) = 0$ if and only if $a=0$, then the trace $\tr$ is called faithful. The functional
\[ \tr : \ell^{1}_{s}(\Lambda\times\Gamma) \to \C, \ \tr(a) = a(0)\]
is a faithful trace on $\ell^{1}(\Lambda\times\Gamma)$.
Naturally, this trace extends to $\mathcal{A}_{s}$ by the isomorphism $I$. Similarly, 
\[ \tr^{\circ} : \ell^{1}_{s}(\Gamma^{\perp}\times \Lambda^{\perp}) \to \C, \ \tr^{\circ} (b) = q\,\vert\alpha \beta\vert \, b(0) \]
defines a faithful trace on $\ell^{1}_{s}(\Gamma^{\perp}\times\Lambda^{\perp})$. Note the normalization of $\tr^{\circ}$.

\begin{lemma} The following equalities hold:
\begin{enumerate}
\item[(i)] $\tr( \lhs{f}{g}) = \tr^{\circ}( \rhs{g}{f} )$, \qquad for all $f,g\in \M{1}{s}(\R\times\Z_{q})$,
\item[(ii)] $\tr(\partial_{j}a) = 0$, $j=1,2$, \qquad for all $a\in \ell^{1}_{s}(\Lambda\times\Gamma)$
\item[(iii)] $\tr^{\circ}(\partial_{j}^{\circ}b) = 0$, $j=1,2$, \qquad for all $b\in \ell^{1}_{s}(\Gamma^{\perp}\times\Lambda^{\perp})$,.
\end{enumerate}
\end{lemma}
\begin{proof}
It is straightforward to establish that
\[ \tr( \lhs{f}{g}  ) = \mathcal{V}_{g}f(0) = \langle f,g\rangle \]
and
\[ \tr^{\circ}(  \rhs{g}{f} )) = q \vert\alpha \beta\vert (q\vert\alpha \beta\vert)^{-1} \overline{\mathcal{V}_{f}g(0)} = \langle f, g\rangle\]
which is (i). The statements (ii) and (iii) are easily verified.
\end{proof}

For any projection $p\in \ell^{1}_{s}(\Lambda\times\Gamma)$, $s\ge 1$, 
its \emph{Connes-Chern number} $c_{1}(p)$ is given by
\[ c_{1}(p) = \frac{1}{2\pi i \, \vert \alpha \beta\vert} \, \tr \big( p [ (\partial_{1}p)(\partial_{2}p) - (\partial_{2}p)(\partial_{1}p)] \big). \]
By general facts \cite{co80} this is an integer number, being the index of a Fredholm operator, that depends only on the class of $p$.
If $p=\lhs{g}{h}$, $g,h\in\M{1}{s}(\R\times\Z_{q})$, $s\ge 1$, then
\[ 
c_{1}(p) =  \frac{2\pi }{ i \, \vert \alpha \beta \vert} \sum_{\nu,\nu'\in\Lambda\times\Gamma} (\lambda'\gamma-\lambda\gamma') \, \mathcal{V}_{h}{g}(\nu) \, \mathcal{V}_{h}{g}(\nu') \, \mathcal{V}_{h}{g}(-\nu-\nu') \overline{\varphi(\nu',\nu'+\nu)}\, \overline{\varphi(\nu,\nu)}, 
\]
where $\nu=(\lambda,\gamma) \in\Lambda\times\Gamma \subset \R \times \widehat{\R}$
and similarly for $\nu'=\lambda', \gamma'$.

We will next show that if $p=\lhs{g}{h}$ and $g$ and $h$ in $\M{1}{s}(\R\times\Z_{q})$, $s\ge 1$, are any pair of functions that generate dual Gabor frames for $L^{2}(\R\times\Z_{q})$ with respect to time-frequency shifts in 
$\Lambda\times\Gamma$,  then $c_{1}(p) = q$.  In order to prove this, we need the following lemma.
%it turns out that it is, in fact, independent on the choice of $p$ and only depends on lattices $\Lambda$ and $\Gamma$ which define the non-commutative torus $\ell^{1}_{s}(\Lambda\times\Gamma)\cong\mathcal{A}_{s}$. Indeed, 

\begin{lemma} \label{le:0803a} Let $g,h\in \M{1}{s}(\R\times\Z_{q})$, $s\ge 1$ be a dual (not necessarily the canonical dual) pair of Gabor frame generators with respect to $\Lambda$ and $\Gamma$. Then
\[ 
\lhs{f_{1}}{\nabla_{j}g} \lhs{h}{f_{2}} + \lhs{f_{1}}{g} \lhs{\nabla_{j}h}{f_{2}} = 0 
\qquad\text{for all} \ \ f_{1},f_{2}\in \M{1}{s}(\R\times\Z_{q}).\]
\end{lemma}
\begin{proof}
By the Wexler-Raz relations for dual generators, we know that $\rhs{g}{h}=1$. Therefore $\partial_{j}^{\circ}\rhs{g}{h}=0$. It follows that for all $f_{1},f_{2}\in\M{1}{}(\R\times\Z_{q})$
\begin{multline*} 
\lhs{f_{1}}{\nabla_{j}g} \lhs{h}{f_{2}} + \lhs{f_{1}}{g} \lhs{\nabla_{j}h}{f_{2}} \\ 
= \lhs{f_{1} \cdot (\rhs{\nabla_{j} g}{h} + \rhs{g}{\nabla_{j} h})}{f_{2}} = \lhs{f_{1} \cdot ( \partial_{j}^{\circ}\rhs{g}{h})}{f_{2}} = 0, 
\end{multline*}
as stated. \end{proof}
\begin{proposition} \label{pr:0103a} If $g,h\in \M{1}{s}(\R\times\Z_{q})$, $s\ge 1$ generate dual frames $\{ \pi(\nu) g\}_{\nu\in\Lambda\times\Gamma}$  and $\{ \pi(\nu) h\}_{\nu\in\Lambda\times\Gamma}$ for $L^{2}(\R\times\Z_{q})$, then, for $p=\lhs{g}{h}$, 
\[ c_{1}(p) = \frac{1}{2\pi i \, \vert \alpha \beta\vert} \, \tr ( p \, [(\partial_{1} p )(\partial_{2}p) - (\partial_{2}p)(\partial_{1}p) ]) 
%= \frac{1}{2\pi i} \tr \big( \lhs{F_{12}g}{g}\big) 
= q. \]
\end{proposition}
\begin{proof} By the Wexler-Raz relations $\rhs{h}{g} = 1$. Hence, for all $f_{1},f_{2}\in \M{1}{}(\R\times\Z_{q})$,  
\begin{align} 
%& \lhs{f_{1}}{g} \lhs{h}{f_{2}} = \lhs{\lhs{f_{1}}{g} h }{f_{2}} = \lhs{ f_{1} \ \rhs{g}{h}}{f_{2}} = \lhs{f_{1}}{f_{2}} \ \ \text{for all} \ \ f_{1},f_{2}\in \M{1}{}(\R\times\Z_{q}) \\
\label{eq:dual-generators-IP-calculation} 
\lhs{f_{1}}{h} \lhs{g}{f_{2}} = \lhs{\lhs{f_{1}}{h} g }{f_{2}} = \lhs{ f_{1} \ \rhs{h}{g}}{f_{2}} = \lhs{f_{1}}{f_{2}}.
\end{align}
In addition, using the linearity of the trace, the cyclic property $\tr(a \cdot b) = \tr(b \cdot a)$ 
for all $a,b\in \mathcal{A}_{s}$ and the result \eqref{eq:0103a} one has the equalities:
\begin{align} 
& \tr ( p \, [(\partial_{1} p )(\partial_{2}p) - (\partial_{2}p)(\partial_{1}p) ]) \nonumber \\
& = \tr \big( \lhs{g}{h} \, [(\partial_{1} \lhs{g}{h} )(\partial_{2}\lhs{g}{h}) - 
(\partial_{2}\lhs{g}{h})(\partial_{1}\lhs{g}{h}) ]\big) \nonumber \\
& \phantom{-} =  \tr \big( \lhs{g}{h} \, [ (\lhs{\nabla_{1}g}{h} +\lhs{g}{\nabla_{1}h}) (\lhs{\nabla_{2}g}{h}+ \lhs{g}{\nabla_{2}h})] \big) \nonumber \\
& \qquad - \tr \big( \lhs{g}{h} \, [(\lhs{\nabla_{2}g}{h}+\lhs{g}{\nabla_{2}h} )(\lhs{\nabla_{1}g}{h}+\lhs{g}{\nabla_{1}h}) ]\big) \nonumber \\
& \phantom{-} = \tr \big( \lhs{g}{h} [\lhs{\nabla_{1}g}{h}\lhs{\nabla_{2}g}{h} + \lhs{\nabla_{1}g}{h}\lhs{g}{\nabla_{2}h} ] \big) \nonumber \\ 
& \qquad + \tr \big( \lhs{g}{h} [\lhs{g}{\nabla_{1}h}\lhs{\nabla_{2}g}{h} + \lhs{g}{\nabla_{1}h}\lhs{g}{\nabla_{2}h} ] \big) \nonumber \\ 
& \qquad - \tr \big( \lhs{g}{h} [\lhs{\nabla_{2}g}{h}\lhs{\nabla_{1}g}{h} + \lhs{\nabla_{2}g}{h}\lhs{g}{\nabla_{1}h} 
\nonumber \\ 
&  \qquad + \tr \big( \lhs{g}{h} [\lhs{g}{\nabla_{2}h}\lhs{\nabla_{1}g}{h} + \lhs{g}{\nabla_{2}h}\lhs{g}{\nabla_{1}h} ] \big) \nonumber \\
& \phantom{-} = \tr\big( \lhs{\nabla_{1}g}{h}\lhs{\nabla_{2}g}{h} + \lhs{g}{h} \lhs{\nabla_{1}g}{\nabla_{2}h} + \lhs{\nabla_{2}g}{\nabla_{1}h} + \lhs{g}{\nabla_{1}h} \lhs{g}{\nabla_{2}h} \big) \nonumber \\
& \qquad - \tr \big( \lhs{\nabla_{2}g}{h}\lhs{\nabla_{1}g}{h} + \lhs{g}{h} \lhs{\nabla_{2}g}{\nabla_{1}h} + \lhs{\nabla_{1}g}{\nabla_{2}h} + \lhs{g}{\nabla_{2}h} \lhs{g}{\nabla_{1}h} \big) \nonumber \\
& \phantom{-} = \tr\big( \lhs{g}{h} \lhs{\nabla_{1}g}{\nabla_{2}h} + \lhs{\nabla_{2}g}{\nabla_{1}h} \big) \nonumber \\
& \qquad - \tr\big(  \lhs{g}{h} \lhs{\nabla_{2}g}{\nabla_{1}h} + \lhs{\nabla_{1}g}{\nabla_{2}h}  \big) . \label{eq:0103e}
\end{align}
Using Lemma \ref{le:0803a} we continue,
\begin{align*} 
\tr ( p \, [(\partial_{1} p )(\partial_{2}p) & - (\partial_{2}p)(\partial_{1}p) ])  \\
& = - \tr\big( \lhs{g}{\nabla_{1}h} \lhs{g}{\nabla_{2}h} \big)  + \tr\big( \lhs{\nabla_{2}g}{\nabla_{1}h} \big) \nonumber \\
& \qquad + \tr\big(  \lhs{g}{\nabla_{2}h} \lhs{g}{\nabla_{1}h}\big) - \tr \big(\lhs{\nabla_{1}g}{\nabla_{2}h}  \big) \\
& = \tr\big( \lhs{\nabla_{2}g}{\nabla_{1}h} \big) \nonumber - \tr \big(\lhs{\nabla_{1}g}{\nabla_{2}h}  \big) \\
& = - \tr\big( \lhs{(\nabla_{1}\nabla_{2} - \nabla_{2}\nabla_{1}) g}{h} \big) \\
& = - \tr\big( \lhs{ F_{12} g}{h}\big) = - \langle F_{12} g,h\rangle = 2\pi i \, \langle g,h\rangle = 2\pi i q \vert\alpha \beta\vert.
\end{align*}
In the last steps we used that $\tr\big( \lhs{f_{1}}{\nabla_{j}f_{2}} \big) = -\tr\big( \lhs{\nabla_{j}f_{1}}{f_{2}} \big)$ for all $f_{1},f_{2}\in\M{1}{s}(\R\times\Z_{q})$ and $j=1,2$. The  very last equality is due to the Wexler-Raz relations.
\end{proof}

\section{An energy functional for  projections}

Let now $\projection$ be the set of all projections $p\in \ell^{1}_{s}(\Lambda\times\Gamma)$. For $s\ge 1$ we define the \emph{energy-functional}
\[ \energy : \projection \to \R^{+}_{0} , \qquad \energy(p) = \frac{1}{4\pi \, \vert \alpha \beta \vert} \, \tr\big( (\partial_{1}p)^{2} + (\partial_{2}p)^{2}\big) .\]
The energy-functional takes non-negative values: the self-adjointness of all sequences $p\in \projection$ together with the fact that $\tr( p^{*}p)\ge 0$, for all $p\in \ell^{1}_{s}(\Lambda\times\Gamma) $ implies that
\[ 0 \le \tr\big( (\partial_{1}p)^{*} (\partial_{1}p) + (\partial_{2}p)^{*}(\partial_{2}p)\big) = \tr\big( (\partial_{1}p)^{2} + (\partial_{2}p)^{2}\big).\]

The following shows that there is an interesting relationship between the energy-functional $E(p)$ and the Connes-Chern number $c_{1}(p)$ of \emph{projections} in $\ell^{1}(\Lambda\times\Gamma)$.

\begin{lemma} \label{le:0804a} The energy-functional is bounded from below by the Connes-Chern-number:
\[
\energy(p) \ge \vert c_{1}(p) \vert \qquad \mbox{for all} \ \ p\in\projection, \, s\ge 1. 
\]
If $p$ satisfies either of the two self-duality or anti-self duality equations, 
\begin{align*} & (\partial_{1} p + i \partial_{2} p) \, p = 0,  
\qquad \mbox{or} \quad p \, (\partial_{1} p - i \partial_{2} p) = 0,\\
& (\partial_{1} p - i \partial_{2} p) \, p = 0 ,
\qquad \mbox{or} \quad p \, (\partial_{1} p + i \partial_{2} p)  = 0, 
\end{align*}
then $\energy(p) = \vert c_{1}(p) \vert $.
\end{lemma}
\begin{proof}
Since $p^{2}=p$ one has $\partial_{j} p = \partial_{j} (p^{2}) = (\partial_{j} p)p + p (\partial_{j} p)$ for $j=1,2.$
This implies that
\[ (\partial_{1} p)^{2} + (\partial_{2}p)^{2}  =  (\partial_{1} p)^{2} \, p + (\partial_{1} p)\, p \, (\partial_{1}p ) + (\partial_{2}p)^{2} \, p+ (\partial_{2}p) \, p \, (\partial_{2} p) . \]
Applying the trace and using its cyclic property, we find
\begin{align*} \tr( (\partial_{1}p)^{2} + (\partial_{2}p)^{2}) & = \tr\Big( (\partial_{1} p)^{2} \, p \Big) + \tr \Big( (\partial_{1} p)\, p \, (\partial_{1} p) \Big) + \tr\Big( (\partial_{2}p)^{2} \, p \Big) + \tr \Big( (\partial_{2}p) \, p \, (\partial_{2} p) \Big) \\
& = 2  \, \tr \Big( p \, [ (\partial_{1}p)^{2} + (\partial_{2}p)^{2} ]\Big).
\end{align*}
This shows that $2 \pi \, \vert \alpha \beta\vert \, \energy(p) = \tr \Big( p \, [ (\partial_{1}p)^{2} + (\partial_{2}p)^{2} ]\Big)$.
Then, the positivity of the trace gives, 
\begin{align}
0 & \le \tr\Big( \big([\partial_{1} p + i \partial_{2} p] \, p\big)^{*} \big([\partial_{1} p + i \partial_{2} p] \, p\big)\Big) \nonumber \\
& = \tr\Big( p^{*} (\partial_{1} p)^{*} (\partial_{1} p) \, p + i \, p^{*} (\partial_{1}p)^{*} (\partial_{2} p) \, p -i \, p^{*} (\partial_{2} p)^{*} (\partial_{1} p) \, p + p^{*} (\partial_{2}p)^{*} (\partial_{2}p) \, p \Big) \nonumber \\
& = \tr\Big( p \, \big[ (\partial_{1}p)^{2} +  (\partial_{2}p)^{2}) \big] \Big) + i \, \tr \Big( p \, \big[ (\partial_{1}p) \, (\partial_{2} p) - (\partial_{2} p) \, (\partial_{1} p) \big]\Big). \label{eq:2802a}
\end{align}
Similarly, one establishes that
\begin{align}  0 & \le \tr\Big( \big([\partial_{1} p - i \partial_{2} p] \, p\big)^{*} \big([\partial_{1} p - i \partial_{2} p] \, p\big)\Big) \nonumber \\
& = \tr\Big( p \, \big[ (\partial_{1}p)^{2} +  (\partial_{2}p)^{2}) \big] \Big) - i \, \tr \Big( p \, \big[ (\partial_{1}p) \, (\partial_{2} p) - (\partial_{2} p) \, (\partial_{1} p) \big]\Big). \label{eq:2802b} \end{align}
Combining \eqref{eq:2802a} and \eqref{eq:2802b} yields the inequality
\[ 2\pi \, \vert \alpha \beta \vert \,\energy(p) = \tr \big( p \, [ (\partial_{1}p)^{2} + (\partial_{2}p)^{2} ]\big) \ge \vert \tr \big( p \, \big[ (\partial_{1}p) \, (\partial_{2} p) - (\partial_{2} p) \, (\partial_{1} p) \big]\big)\vert = 2\pi \, \vert \alpha \beta \, c_{1}(p) \vert. \]
Since $\tr(a^{*} a)=0$ if and only if $a=0$ it is clear that equality holds if either of the two 
equations above are satisfied. 
\end{proof}

\subsection{An energy functional for Gabor frame generators}

Let $\mathcal{G}(\Lambda\times\Gamma)$ denote the set of all functions $g\in \M{1}{s}(\R\times\Z_{q})$, $s\ge 1$,
that generate a Gabor frame $\{\pi(\nu)g\}_{\nu\in\Lambda\times\Gamma}$ for $L^{2}(\R\times\Z_{q})$. If we apply the energy functional $E$ from the previous section to the projection $p=\lhs{g}{S^{-1}_{g}g}$, $g\in \mathcal{G}(\Lambda\times\Gamma)$, then one finds that
\[ E(p) = \frac{\pi}{\vert \alpha \beta \vert} \sum_{\nu\in\Lambda\times\Gamma} 
(\lambda^{2}+\gamma^{2}) \, \vert \mathcal{V}_{g} (S_{g}^{-1}g)(\nu) \vert^{2}.\]

One then has an energy functional for Gabor frame generators $g\in \M{1}{s}(\R\times\Z_{q})$, $s\ge 1$, 
\begin{equation} \label{eq:gabor-functional} E : \mathcal{G}(\Lambda\times\Gamma) \to \R_{0}^{+} , 
\qquad E(g) = \frac{\pi}{\vert \alpha \beta \vert} \sum_{\nu\in\Lambda\times\Gamma} (\lambda^{2}+\gamma^{2}) \, \big\vert \big\langle g, \pi(\nu)  S_{g}^{-1}g\big\rangle \big\vert^{2}\end{equation}
It follows from Lemma \ref{le:0804a} that this is bounded from below by $q$. Moreover, we know that the minimum value $E(g) = q$ is obtained for those functions $g$, where $p=\lhs{g}{S^{-1}_{g}g}$ satisfies either of the two equations in Lemma \ref{le:0804a}. As it turns out, 
the duality principle for Gabor frames allows us to find minimisers of this functional.

As we did earlier when we described the duality principle for Gabor frames, we let $W$ be the closure of $\mbox{span}\{\pi^{\circ}(\nu^{\circ})g\}$ in $\M{1}{s}(\R\times\Z_{q})$.

\begin{theorem} \label{th:soliton-solution} 
If $g\in \mathcal{G}(\Lambda\times\Gamma)$ satisfies either of the following conditions,
\begin{enumerate} 
\item[(i)] $(\nabla_{1}+i\nabla_{2})g\in W$,
\item[(ii)] $(\nabla_{1}-i\nabla_{2})g\in W$,
\end{enumerate}
then $g$ minimizes the energy functional \eqref{eq:gabor-functional}, that is,
\[ E(g) = \frac{\pi}{\vert \alpha \beta \vert} \sum_{\nu\in\Lambda\times\Gamma} (\lambda^{2}+\gamma^{2}) \, \big\vert \big\langle g, \pi(\nu)  S_{g}^{-1}g\big\rangle \big\vert^{2} = q.\]

\end{theorem}
\begin{proof} Let $p=\lhs{g}{S^{-1}_{g}g}$. We will show that if (i) is satisfied, then 
$(\partial_{1}p + i \partial_{2}p) p = 0$ 
and if (ii) holds then
$(\partial_{1}p - i \partial_{2}p) p = 0$.
In either case Lemma \ref{le:0804a} implies that $p$ minimizes the energy functional $\energy$. 
Using the fact that $g$ and $S^{-1}g$ are dual frame generators, it is a straightforward computation 
with the help of \eqref{eq:dual-generators-IP-calculation} to show that
\begin{equation} \label{eq:0804c}
(\partial_{1}p \pm  i \partial_{2}p) p = \lhs{(\nabla_{1}\pm i\nabla_{2}) g}{S^{-1}g} + \lhs{g \, \rhs{(\nabla_{1}\mp i \nabla_{2}) S^{-1}g}{g}}{S^{-1}g}
\end{equation}
Since $\rhs{S^{-1}g}{g}=1$ it follows that
\[ \rhs{\nabla_{j}S^{-1}g}{g} + \rhs{S^{-1}g}{\nabla_{j}g} = \partial_{j} \rhs{S^{-1}g}{g} = 0, \quad  j=1,2.\]
Therefore
\[ \rhs{(\nabla_{1}\mp i \nabla_{2}) S^{-1}g}{g} = -\rhs{ S^{-1}g}{(\nabla_{1}\pm i \nabla_{2})g}.\]
With this we continue the calculation in \eqref{eq:0804c} and establish that
\begin{equation} \label{eq:0804d}(\partial_{1}p \pm  i \partial_{2}p) p = \lhs{(\nabla_{1}\pm i\nabla_{2}) g}{S^{-1}g} - \lhs{g \, \rhs{ S^{-1}g}{(\nabla_{1}\pm i \nabla_{2})g}}{S^{-1}g}.\end{equation}
The duality principle for Gabor frames yields that any function $f\in W$ can be written as 
\[ f =  g \, \rhs{S^{-1}g}{f}. \]
The assumption (i) and (ii) ensure that we have such a representation available for the function $(\nabla_{1}\pm i \nabla_{2})g$. Therefore
\[ g \, \rhs{ S^{-1}g}{(\nabla_{1}\pm i \nabla_{2})g} = (\nabla_{1}\pm i \nabla_{2})g.\]
Using this in \eqref{eq:0804d} yields that
\[ (\partial_{1}p \pm  i \partial_{2}p) p = \lhs{(\nabla_{1}\pm i\nabla_{2}) g}{S^{-1}g} - \lhs{(\nabla_{1}\pm i \nabla_{2})g}{S^{-1}g} = 0. \]
Thus concluding the proof. 
\end{proof}

Note that conditions (i) and (ii) in Theorem \ref{th:soliton-solution} are first order differential equations. The following lemma details a solution to these \emph{soliton equations} for any \emph{topological charge} 
$q$, provided the parameters $\alpha,\beta,r$ and $s$ defining the lattices $\Lambda$ and $\Gamma$ are suitably chosen.

\begin{lemma} \label{le:soliton-diff-solution} If $\alpha,\beta,r,s$ and $q$ are such that
\[ (\alpha\beta q^{2})^{-1} + r^{\circ}s^{\circ}/q \in\Z \quad \text{and} \quad \vert\alpha\beta\vert \, q < 1, \]
then, for any non-zero $c\in\C$ and any $\lambda\in \C$ the Schwartz function 
\[ g(x,k) = c \, e^{-\pi x^{2}-i \lambda x}, \ x\in\R, k\in\Z_{q},\]
minimizes the energy functional $E$ for the non-commutative torus with these parameters.
\end{lemma}
\begin{proof}
It is straightforward to show that $(\nabla_{1}+i\nabla_{2}) g = \lambda g$, which then implies that $(\nabla_{1}+i\nabla_{2}) g \in W$. The function $g(\,\cdot\,,k)$, for $k\in\Z_{q}$, belongs to the Schwartz space, hence, in particular to all the weighted modulation spaces $\M{1}{s}(\R\times\Z_{q})$ for all $s\ge 0$. Furthermore, the knowledge that the Gabor system $\{ E_{m\beta q}T_{n\alpha} \tilde{g}\}_{m,n\in\Z}$, where $\tilde{g}$ is the generalized Gaussian $g(\cdot, k)$, is a frame for $L^{2}(\R)$ if and only if $\vert \alpha\beta\vert q < 1$ together with Lemma \ref{le:from-1-to-q-frames} shows that $g\in \mathcal{G}(\Lambda\times\Gamma)$. Hence the result follows from Theorem \ref{th:soliton-solution}.
\end{proof}

It is unknown whether there are other functions besides the Gaussian that satisfy the assumptions of Theorem \ref{th:soliton-solution}.

\section{The continuous picture -- the Moyal plane}
We close with the construction of solitons of general topological charge for the Moyal plane $\mathcal{A}$ which extends the results in \cite{dalalu15}. 
So far we have considered Gabor systems of the form $\{\pi(\nu)g\}_{\nu\in\Lambda\times\Gamma}$, where $\Lambda$ and $\Gamma$ are lattices in the time and frequency domain. The presented theory also works for continuous Gabor systems in $L^{2}(\R\times\Z_{q})$, where one takes $\Lambda\times\Gamma$ to be the \emph{entire} time-frequency plane. That is, for some $g\in \M{1}{s}(\R\times\Z_{q})$ we consider the Gabor system of the form 
\[ \{E_{\omega,c}T_{x,l} \, g \, :  \, (x,l,\omega,c) \in\R\times\Z_{q}\times\widehat{\R}\times\widehat{\Z}_{q}\}.\]
Such a Gabor system is a frame for $L^{2}(\R\times\Z_{q})$ if there exists constants $A,B>0$ such that
\[ A \, \Vert f \Vert_{2}^{2} \le \sum_{c,l\in\Z_{q}} \int_{\R^{2}} \vert \langle f, E_{\omega,c}T_{x,l} g \rangle \vert^{2} \, \dd(x,\omega) \le B \, \Vert f \Vert_{2}^{2} \ \ \text{for all} \ \ f\in L^{2}(\R\times\Z_{q}).\]
The theory of continuous Gabor frames is not as intricate as the one of discrete Gabor frames because the Moyal identity states that all functions $g\in\M{1}{s}(\R\times\Z_{q})$ (in fact, all functions in $L^{2}(\R\times\Z_{q})$) generate a continuous Gabor frame with bounds $A=B= q\, \Vert g \Vert_{2}^{2}$. Specifically,
\begin{equation} \label{eq:moyal} 
\sum_{c,l\in\Z_{q}} \int_{\R^{2}} \vert \langle f, E_{\omega,c}T_{x,l} g \rangle \vert^{2} \, \dd(x,\omega)  = q \, \Vert g \Vert_{2}^{2} \, \Vert f \Vert_{2}^{2} \quad \text{for all} \ \ f\in L^{2}(\R\times\Z_{q}). 
\end{equation}

Further, we need the vector space $\M{1}{s}(\R\times\Z_{q}\times\widehat{\R}\times\widehat{\Z}_{q})$.
It becomes an involutive Banach algebra under the twisted convolution and twisted involution given by
\begin{align*} 
(k_{1}\twist k_{2})(\nu) & = \sum_{\Z_{q} \times \Z_{q}}\int_{\R^{2}} k_{1}(\nu') k_{2}(\nu-\nu') \varphi(\nu',\nu-\nu') \, \dd\nu', \quad \nu\in \R\times\Z_{q}\times\widehat{\R}\times\widehat{\Z}_{q}\\
k^{*}(\nu) & = \varphi(\nu,\nu) \overline{k(-\nu)}, \ \ \nu\in \R\times\Z_{q}\times\widehat{\R}\times\widehat{\Z}_{q}.\end{align*}
One can show that the map
\[ I : k\mapsto \sum_{\Z_{q} \times \Z_{q}} \int_{\R^{2}} k(\nu) \, \pi(\nu) \, \dd\nu \]
is an isomorphism from $\M{1}{s}(\R\times\Z_{q}\times\widehat{\R}\times\widehat{\Z}_{q})$ onto the involutive Banach algebra
\[ \mathcal{A}_{s} = \big\{ T : \M{1}{s}(\R\times\Z_{q})\to \M{1}{s}(\R\times\Z_{q}), \ T =\sum_{\Z_{q} \times \Z_{q}} \int_{\R^{2}} k(\nu) \, \pi(\nu) \, \dd\nu, \ k\in \M{1}{s}(\R\times\Z_{q}\times\widehat{\R}\times\widehat{\Z}_{q}) \big\}.\]
Indeed, $\mathcal{A}_{s}$ is an involutive Banach algebra for the norm $\Vert T\Vert_{A_{s}} = \Vert k \Vert_{L^{1}_{s}}$, composition of operators and the involution of $T\in \mathcal{A}_{s}$ being its $L^{2}$-Hilbert space adjoint $T^{*}$.
We define the left action of a function $k\in L^{1}_{s}(\R\times\Z_{q}\times\widehat{\R}\times\widehat{\Z}_{q})$ on a function $f\in \M{1}{s}(\R\times\Z_{q})$ by
\[ k \cdot f = I(k) f = \sum_{\Z_{q} \times \Z_{q}} \int_{\R^{2}} k(\nu) \, \pi(\nu) f\, \dd\nu.\]
The $\M{1}{s}(\R\times\Z_{q}\times\widehat{\R}\times\widehat{\Z}_{q})$-valued inner-product is defined by
\[ \lhs{ \cdot }{ \cdot } : \M{1}{s}(\R\times\Z_{q})\times \M{1}{s}(\R\times\Z_{q}) \to 
\M{1}{s}(\R\times\Z_{q}\times\widehat{\R}\times\widehat{\Z}_{q}), \quad \lhs{f}{g} = \mathcal{V}_{g}f.\]
One can verify analogous of properties \eqref{eq:inner-p-rule}. That is, 
for all $f,g\in \M{1}{s}(\R\times\Z_{q})$ and for all $k\in \M{1}{s}(\R\times\Z_{q}\times\widehat{\R}\times\widehat{\Z}_{q})$ one has, 
\begin{align} \label{eq:inner-p-rule-cont} 
k \twist \lhs{f}{g} &= \lhs{k \cdot f}{g} \ , \quad \lhs{f}{g}\twist k^{*} = \lhs{f}{k \cdot g} \ , 
\quad (\lhs{f}{g})^{*} = \lhs{g}{f} \nn \\
\lhs{f}{f} &\ge 0 \quad\text{and} \quad \lhs{f}{f} = 0 \ \ \Leftrightarrow \ \ f = 0.
\end{align}
The associated enveloping $C^*$-algebra is the Moyal plane $\mathcal{A}$ represented by compact operators on $L^2(\R\times\Z_{q})$, see \cite{dalalu15} for the scalar case $q=1$.

Since we are really considering the short-time Fourier transform, that is Gabor systems with time-frequency shifts along the entire time-frequency plane, the annihilator is just a single point. 
We therefore have $\mathcal{A}_{s}^{\circ} \simeq \C$, where $\C$ takes the role of $\ell^{1}_{s}(\Gamma^{\perp}\times\Lambda^{\perp})$ from earlier:
\[ \mathcal{A}_{s}^{\circ} = \big\{ T: \M{1}{s}(\R\times\Z_{q})\to \M{1}{s}(\R\times\Z_{q}) \, : \, T = b, \, b\in \C \big\}.\]
%and clearly $\mathcal{A}_{s}^{\circ} \cong \C$, where $\C$ takes the role of $\ell^{1}_{s}(\Gamma^{\perp}\times\Lambda^{\perp})$ from earlier.
The right action of elements $b\in \C$ on a function $f\in \M{1}{s}(\R\times\Z_{q})$ is given by
\[ f \cdot b = I(b)f = bf.\]
And the $\C$-valued inner-product on $\M{1}{s}(\R\times\Z_{q})$ is now  
\[ \rhs{ \cdot }{ \cdot } : \M{1}{s}(\R\times\Z_{q})\times\M{1}{s}(\R\times\Z_{q})\to \C \, , \quad \lhs{f}{g} = q\, \langle g,f\rangle.\]

A variation of \cite[Prop. 5.1]{dalalu15} shows that $\M{1}{s}(\R\times\Z_{q})$ is a singly generated projective module over $\mathcal{A}$, that is $\M{1}{s}(\R\times\Z_{q})$ is a line bundle over the Moyal plane and the 
Moyal plane algebra  
$\mathcal{A}_s$ is Morita equivalent to $\mathbb{C}$.

\begin{proposition}
The algebra $\mathcal{A}_s$ is Morita equivalent to $\mathbb{C}$ and 
$\M{1}{s}(\R\times\Z_{q})$ is an equivalence bimodule. Any $g\in \M{1}{s}(\R\times\Z_{q})$ with $\|g\|_{L^2(\R\times\Z_{q})}=1$ generates $\M{1}{s}(\R\times\Z_{q})$, that is 
for any $f\in \M{1}{s}(\R\times\Z_{q})$ we have $f=\lhs{f}{g}g$.
\end{proposition}

Since the annihilator of the full time-frequency plane only consists of one point the main results of the Gabor frame theory reduce to well known facts:

The \textbf{fundamental identity of Gabor analysis}:\\
just states that for all $f,g,h\in \M{1}{s}(\R\times\Z_{q})$ one has that
\begin{equation} \label{eq:janssen-cont} 
\lhs{f}{g} \cdot h = f \cdot \rhs{g}{h}
\end{equation}
that is, $\displaystyle \sum_{\Z_{q}}\int_{\R^{2}} \langle f, \pi(\nu) g \rangle \,  \pi(\nu) h \, \dd\nu = q\, \langle h, g\rangle \, f$, a version of the Moyal identity \eqref{eq:moyal}.

%In the time-frequency community \eqref{eq:janssen} yields a useful way of expressing the Gabor frame operator,
%\[ S_{g} f = f \, \rhs{g}{g}.\]
%In non-commutative geometry the validity of \eqref{eq:janssen} implies that the non-commutative torus $\mathcal{A}_{s}$ and $\mathcal{A}_{s}^{\circ}$ are Morita-equivalent, which is a far deeper.

The \textbf{Wexler-Raz biorthogonality relations}: \\ 
just characterises when two functions $g,h\in \M{1}{s}(\R\times\Z_{q})$ generate dual continuous Gabor frames. This is the case if and only if
\[ \rhs{g}{h} = 1, \quad \text{that is} \ \langle h, g\rangle = q^{-1}, \]
and the construction of a pair of generators for dual continuous Gabor frames is trivial. 

The \textbf{duality principle for Gabor frames}:\\
just becomes the simple statement, that for $g$ and $h$ such that $\rhs{g}{h}=1$, we have 
\begin{equation} \label{eq:riesz-rep-cont} 
f = g \cdot \rhs{h}{f} \qquad \text{for all} \ \ f\in W = \mbox{span}\, g. 
\end{equation}

As in Section \ref{se:dcc}, we have covariant derivatives on the line bundle $\M{1}{s}(\R\times\Z_{q})$:
\begin{itemize}
\item[] $\nabla_{1}: \M{1}{s+1}(\R\times\Z_{q}) \to \M{1}{s}(\R\times\Z_{q}), 
\quad \nabla_{1}f(\, \cdot \, , k) = 2\pi \, i \, M f(\, \cdot \, , k), \quad k\in \Z_{q}$,
\item[] $\nabla_{2}: \M{1}{s+1}(\R\times\Z_{q}) \to \M{1}{s}(\R\times\Z_{q}), \quad
\nabla_{2}f(\, \cdot \,, k) = D f(\, \cdot\, , k), \quad k\in \Z_{q}$.
\end{itemize}
On $L^{1}_{s}(\R\times\Z_{q}\times\widehat{\R}\times\widehat{\Z}_{q})$, for $s\ge 0$, 
we define derivations $\partial_{1}$ and $\partial_{2}$, 
$$
\partial_{j}: L^{1}_{s+1}(\R\times\Z_{q}\times\widehat{\R}\times\widehat{\Z}_{q}) 
\to L^{1}_{s}(\R\times\Z_{q}\times\widehat{\R}\times\widehat{\Z}_{q}), \quad j+1,2 
$$
given, for $(x,l,\omega,c)\in \R\times\Z_{q}\times\widehat{\R}\times\widehat{\Z}_{q})$ by, 
$$
(\partial_{1}k)(x,l,\omega,c) = 2\pi i x \, k(x,l,\omega,c), \qquad  
(\partial_{2}k)(x,l,\omega,c) = 2 \pi i \omega \, k(x,l,\omega,c) .
$$
On $\C$ there are just the trivial derivations:
$$\
\partial_{j}^{\circ}: \C\to\C, \qquad \partial_{j}^{\circ}b = 0, \ b\in \C , \quad j=1,2 .
$$
Clearly, the isomorphisms $I$ and $I^{\circ}$ between $\M{1}{s}(\R\times\Z_{q}\times\widehat{\R}\times\widehat{\Z}_{q})$ and $\mathcal{A}_{s}$ and between $\C$ and $\mathcal{A}_{s}^{\circ}$, respectively, allow us define the derivations on $\mathcal{A}_{s}$ and $\mathcal{A}_{s}^{\circ}$, such that
\[ \partial_{j} : \mathcal{A}_{s+1} \to \mathcal{A}_{s} \ \ \text{and} \ \ \partial_{j}^{\circ} : \mathcal{A}_{s+1}^{\circ} \to \mathcal{A}_{s}^{\circ} \ \ \text{for} \ \ j=1,2.\]

In parallel with what happens for the discrete case, 
from the definition of the operators $\nabla_{j}$ and the derivations $\partial_{j}$, $j=1,2$ one establishes the Leibniz rule,
\begin{align}\label{eq:0103b-cont}
\nabla_{j}( a \cdot  f ) = (\partial_{j} a) \cdot f + a\cdot \nabla_{j} f \quad\text{for all} \, f\in \M{1}{1}(\R\times\Z_{q}), \ a\in L^{1}_{s}(\R\times\Z_{q}\times\widehat{\R}\times\widehat{\Z}_{q})
\end{align}
and that there is compatibility with the $\M{1}{s}(\R\times\Z_{q}\times\widehat{\R}\times\widehat{\Z}_{q})$ 
valued inner-product $\lhs{ \cdot }{ \cdot }$, 
\begin{align} \partial_{j} (\lhs{f}{g} ) =  
\lhs{\nabla_{j} f}{g} + \lhs{f}{\nabla_{j} g} \qquad \text{for all} \ \ f,g\in \M{1}{1}(\R\times\Z_{q}) \label{eq:0103a-cont} .\end{align}

Combining these two equationswe find that, for all $f,g,h\in \M{1}{1}(\R\times\Z_{q})$,
\[ \nabla_{j}(\lhs{f}{g} \, h) = \lhs{\nabla_{j} f}{g} \, h + \lhs{f}{\nabla_{j} g} \, h + \lhs{f}{g}  \, \nabla_{j} h \qquad j=1,2.\]
Similar statements hold with $\partial^{\circ}_{j}$
 and $\rhs{ \cdot }{ \cdot }$ instead of $\partial_{j}$ and $\lhs{\,\cdot\,}{\,\cdot\,}$.

The curvature of the covariant derivatives is, as earlier, given by
\[ F_{12} = \nabla_{1}\nabla_{2} - \nabla_{2}\nabla_{1}\]
and is the linear and bounded operator computed to be, 
\[ F_{12}: \M{1}{s}(\R\times\Z_{q}) \to \M{1}{s}(\R\times\Z_{q}), \quad F_{12} f = -2\pi i \, \mbox{Id}.\]
The functional
\[ 
\tr : \M{1}{s}(\R\times\Z_{q}\times\widehat{\R}\times\widehat{\Z}_{q}) \to \C, \quad \tr(k) = k(0),
\]
is a faithful trace on $\M{1}{s}(\R\times\Z_{q}\times\widehat{\R}\times\widehat{\Z}_{q})$.
Naturally, this trace extends to $\mathcal{A}_{s}$ by the isomorphism $I$. 
Similarly, 
\[ \tr^{\circ} : \C \to \C, \quad \tr^{\circ} (b) = q^{-1} b, \ b\in\C \]
defines a faithful trace on $\C$. 

For all $f,g\in\M{1}{s}(\R\times\Z_{q})$, $k\in L^{1}(\R\times\Z_{q}\times\widehat{\R}\times\widehat{\Z}_{q})$, $b\in\C$ and $j=1,2$,
\[ \tr(\partial_{j}k) = 0, \quad \tr^{\circ}(\partial_{j}^{\circ}b) = 0 \quad \mbox{and} \quad \tr( \lhs{f}{g}  ) = \tr^{\circ}( \rhs{g}{f} ). \]

If $p\in L^{1}_{s}(\R\times\Z_{q}\times\widehat{\R}\times\widehat{\Z}_{q})$, $s\ge 1$, is a projection, $p^2=p=p^*$, 
its Connes-Chern number is now given by
\[ c_{1}(p) = \frac{q^{2}}{2\pi i} \, \tr\big( p [ (\partial_{1}p)(\partial_{2}p) - (\partial_{2}p)(\partial_{1}p)] \big). \]
Note the difference when compared to the one used earlier for discrete Gabor systems.
In a way similar to the prove of Proposition \ref{pr:0103a}, one shows the following.
\begin{proposition}
If $g,h\in \M{1}{s}(\R\times\Z_{q})$, $s\ge 1$, 
generate dual continuous Gabor frames, that is, $\langle h,g\rangle = q^{-1}$, then, for $p=\lhs{g}{h}$ one finds
\[ c_{1}(p) = q. \]
\end{proposition}

Let now $\projection$ be the set of all projections $p\in \M{1}{s}(\R\times\Z_{q}\times\widehat{\R}\times\widehat{\Z}_{q})$. For $s\ge 1$, we define the \emph{energy-functional}
\[ \energy : \projection \to \R^{+}_{0} , \quad \energy(p) = \frac{q^{2}}{4\pi} \, \tr\big( (\partial_{1}p)^{2} + (\partial_{2}p)^{2}\big) .\]
Similarly to the discrete case, one shows that energy-functional is bounded from below by the Connes-Chern-number. Specifically,
\begin{equation} \label{eq:2802c-continuous} \energy(p) \ge \vert c_{1}(p) \vert \ \ \mbox{for all} \ \ p\in\projection, \, s\ge 1. \end{equation}
And, if $p$ satisfies either of the two equations
\begin{align*} 
& (\partial_{1} p + i \partial_{2} p) \, p = 0,  \quad \mbox{or} \quad p \, (\partial_{1} p - i \partial_{2} p) = 0,\\
& (\partial_{1} p - i \partial_{2} p) \, p = 0 \quad \mbox{or} \quad p \, (\partial_{1} p + i \partial_{2} p)  = 0,
\end{align*}
then $\energy(p) = \vert c_{1}(p) \vert $.
Let $\mathcal{G}$ denote the set of all functions $g\in \M{1}{s}(\R\times\Z_{q})$, $s\ge 1$,
which generate a continuous Gabor frame $\{\pi(\nu)g\}_{\nu\in\R\times\Z_{q}\times\widehat{\R}\times\widehat{\Z}_{q}}$ for $L^{2}(\R\times\Z_{q})$. Note that this set comprises \emph{all} functions in $\M{1}{s}(\R\times\Z_{q})$. If we apply the energy functional $E$ to the projection $p=\lhs{g}{S^{-1}_{g}g}=\Vert g \Vert_{2}^{-2} \lhs{g}{g}$, $g\in \mathcal{G}$, we find that
\begin{equation} \label{eq:cont-gabor-functional} E : \mathcal{G} \to \R_{0}^{+} , \ E(g) = \frac{q^{2}\,\pi}{\Vert g \Vert_{2}^{2}} \sum_{\Z_{q} \times \Z_{q}} \int_{\R^{2}} (x^{2}+\omega^{2}) \, \vert \mathcal{V}_{g} g(\nu) \vert^{2} \, \dd\nu,\end{equation}
where $\nu=(x,l,\omega,c)\in \R\times\Z_{q}\times\widehat{\R}\times\widehat{\Z}_{q}$. 
This functional is bounded below by $q$. 

Next one shows the analogue of Theorem \ref{th:soliton-solution}.
\begin{proposition}
The \emph{unique} solution among all functions $g\in \M{1}{s}(\R\times\Z_{q})$ to either of the conditions
\[ 
(\nabla_{1}\pm i\nabla_{2})g\in \mathrm{span}\,g
\]
is the generalized Gaussian
\[ g(x,k) = c_{k} \, e^{-\pi x^{2}-i \lambda x}, \ x\in\R, k\in\Z_{q}, \ \ \{c_{k}\}\in\C^{q}.\]
\end{proposition}
\noindent
Hence, for the continuous Gabor transform and its associated soliton equation, the only solution that we can produce is the generalized Gaussian. Note the major difference here compared to the one for the discrete Gabor frames considered before: there one needs $(\nabla_{1}\pm i\nabla_{2})g$ to lie in the space spanned by all time-frequency shifts of the adjoint lattice of the generator $g$. It is therefore reasonable to conjecture the existence of more functions besides the Gaussian that solve the soliton equation for discrete Gabor frames as described in the previous sections.

\vspace{0.5cm}
\noindent\textbf{Acknowledgements.}
The second author gratefully acknowledges that this work was 
carried out during the tenure of an ERCIM 'Alain Bensoussan` Fellowship
Programme at NTNU.

%\bibliography{nhgbib}

\end{document}